\documentclass[journal,twocolumn]{IEEEtran}

\usepackage{cite}
\usepackage{braket}
\usepackage{amsmath,amssymb,amsfonts}
\usepackage{algorithmic}
\usepackage{graphicx}
\usepackage{textcomp}
\def\BibTeX{{\rm B\kern-.05em{\sc i\kern-.025em b}\kern-.08em
    T\kern-.1667em\lower.7ex\hbox{E}\kern-.125emX}}

\usepackage{amsfonts}
\usepackage{amsmath}
\usepackage{amssymb}
\usepackage{array}
\usepackage{mathrsfs}
\usepackage{leftidx}
\usepackage{graphicx}
\usepackage{epstopdf}
\usepackage{dcolumn}
\usepackage{bm}
\usepackage[usenames]{color}
\usepackage{colortbl,booktabs}
\usepackage{subfigure}
\usepackage{cite}
\usepackage{multirow}
\usepackage{amsthm}

\usepackage{cases}
   
\usepackage[switch]{lineno} 

\usepackage{ulem}

\begin{document}
\title{Quantum coherent feedback control of an $N$-level atom with multiple excitations\thanks{This work is partially financially supported by Innovation Program for Quantum Science and Technology 2023ZD0300600, Guangdong Provincial Quantum Science Strategic Initiative (No. GDZX2203001), Hong Kong Research Grant Council (RGC) under Grant No. 15213924, National Natural Science Foundation of China under Grants No. 62173288.}}
\author{Haijin Ding, Guofeng Zhang
\thanks{Haijin Ding is with the Laboratoire des Signaux et Syst\`{e}mes (L2S), CNRS-CentraleSup\'{e}lec-Universit\'{e} Paris-Sud, Universit\'{e} Paris-Saclay, 3, Rue Joliot Curie, 91190, Gif-sur-Yvette, France (e-mail: dhj17@tsinghua.org.cn, haijin.ding@centralesupelec.fr). }
\thanks{Guofeng Zhang is with the Department of Applied Mathematics, The Hong Kong Polytechnic University, Hung Hom, Kowloon,  SAR, China, and The Hong Kong Polytechnic University Shenzhen Research Institute, Shenzhen, 518057, China (e-mail: guofeng.zhang@polyu.edu.hk).}
\thanks{Corresponding author: Guofeng Zhang.}}

\maketitle

\begin{abstract}
The purpose of this paper is to study the dynamics of a quantum coherent feedback network, where an $N$-level atom is coupled with a cavity and the cavity is also coupled with a single or multiple parallel waveguides. When the atom is initialized at the highest energy level, it can emit multiple photons into the cavity, and the photons can be further transmitted to the waveguides and re-interact with the cavity quantum electrodynamics (cavity-QED) system. The transmission of photons in the waveguide can construct a feedback channel with a delay determined by the length of the waveguide. We model the dynamics of the atomic and photonic states of the cavity-QED system as a linear control system with delay. By tuning the control parameters such as the coupling strengths among the atom, cavity and waveguide, the eigenstates of the quantum system can be exponentially stable or unstable, and the exponential stability of the linear quantum control system with delay is related with the generation of single- and multi-photon states. 
Besides, when the cavity-QED system is coupled with multiple parallel waveguides, the emitted photons oscillate among different waveguides and the stability of quantum states is influenced by the feedback loop length and coupling strengths among waveguides.
\end{abstract}

\section{Introduction}\label{Sec:Introduction}
Quantum feedback control has attracted increasing attention due to its wide applications in quantum information processing (QIP)~\cite{wiseman2009quantum,nielsen2001quantum,zhang2012quantum}, communication~\cite{cooney2016strong,salathe2018low,yamamoto2016quantum} and engineering~\cite{Zhangjing2017,vahlbruch2007quantum}. 
On one hand, quantum feedback can be realized by means of classical controllers using the measurement information of a quantum system; this is called quantum measurement feedback~\cite{MeasurementFeedbackRef22,CoherentQFLloyld}. On the other hand, the quantum system can also be regulated by a   controller which itself is a quantum system; this is called the coherent feedback control~\cite{CoherentQFLloyld}. For example, to control a trapped ion, another ion can work as a quantum controller, then the system ion and the controller ion can construct a joint system~\cite{CoherentQFLloyld,CiracPRL1995}. By focusing the light separately on the system ion or the controller ion, the system ion can be modulated by the controller ion via the vibration modes between them~\cite{CoherentQFLloyld,CiracPRL1995}.

Recently, a most widely used quantum coherent feedback realization is based on quantum systems closed by a waveguide~\cite{ZollerPRL,AntonPRL,photonfeedback,ding2022quantum,TrajMultiPhoton,regidor2021modeling}. For example, an atom can be coupled with a cavity, and the cavity is coupled to a waveguide through two semi-transparent mirrors~\cite{photonfeedback}. In this set-up, the photons in the cavity can be transmitted into the waveguide, and re-interact with the cavity and atom after propagating along the waveguide. Thus, the flying photons in the waveguide can  facilitate a coherent feedback channel. Physically, this quantum coherent feedback network can stabilize the Rabi oscillation in the cavity~\cite{2012Stabilizing,lang1973laser}, and control the number of photons in a cavity by tuning the feedback loop length~\cite{photonfeedback,ding2022quantum}.
A quantum coherent feedback network can also be constructed when an atom is directly coupled with a semi-infinite waveguide. The photon emitted by the atom can be reflected by the terminal mirror of the waveguide, and then the reflected photon can re-excite the atom to construct the feedback loop~\cite{waveguideOneatom}. When there are two or multiple two-level atoms coupled with the waveguide, the reflection of the photons by  atoms can induce multiple feedback loops, which can be used to create entangled quantum states~\cite{ZhangBin} and generate multi-photon states~\cite{FanPRB,TrajMultiPhoton} for quantum state engineering~\cite{pan2012multiphoton,dell2006multiphoton}. In the coherent feedback architectures mentioned above, the quantum control dynamics are in general {\it non-Markovian} because the future evolution of the quantum state is not only influenced by its current state information, but also by its historical information related to the length of the feedback loop involved~\cite{photonfeedback,ding2022quantum}. 

Similar to some classical non-Markovian networks~\cite{liu2010controllability,garate2010equivalence,wu2014suboptimal,huang2022impulsive}, most of these quantum feedback networks can be modeled as a linear control system with single or multiple time delays. When the waveguide is coupled with one cavity~\cite{photonfeedback,ding2022quantum} or atom~\cite{regidor2021modeling}, the quantum dynamics can be modeled as a linear system with single delay. When there are multiple quantum components (i.e., cavities or atoms) coupled with the waveguide~\cite{AntonPRL,ZhangBin,ClusterZoller}, there can be several closed feedback loops established by the photons transmitted in the waveguide among different cavities or atoms. This type of quantum control networks can be modeled  as  multi-agent networks with multiple delays~\cite{ZollerPRL,AntonPRL,ramirez2018single,munz2011delay,lu2017leader}.
By means of control theory for time-delayed linear systems~\cite{fridman2002improved,2002An}, these quantum coherent feedback networks can be analyzed and  engineered to achieve  physically meaningful control performance such as photon generation~\cite{ClusterZoller,shi2021deterministic}, networking~\cite{dongCui2023dynamics,gough2008linear,gough2009quantum} and entanglement creation~\cite{ZhangBin}. Hence, it is practically important to investigate this type of delay-dependent quantum coherent feedback systems from the perspective of systems and control theory.

In these quantum coherent feedback networks closed by the transmission of photons, the photon source can be multiple excited two-level atoms~\cite{DZWW23}, an excited multi-level atom~\cite{pan2017scattering}, or a single two-level atom which is driven repeatedly~\cite{ClusterZoller}. One typical example is given in Fig.~\ref{fig:Nlevel}, where an excited multi-level atom sits in a cavity which is coupled to a waveguide. Photons traveling along the waveguide form a coherent feedback loop. The eigenstate space of this feedback network can be represented in terms of the atomic state as well as the number of photons in the cavity and waveguide. The number of photons in the cavity and waveguide is related to the oscillation properties of the amplitudes of the eigenstates, and this oscillation of amplitudes can be revealed by the stability properties of the quantum feedback network~\cite{ding2022quantum}. For example, when the amplitudes of the states for which the cavity contains photons converge to zero, then all the emitted photons will eventually be in the waveguide.  

Similarly to a classical linear control system with delays, the stability of the quantum coherent feedback  network  in Fig.~\ref{fig:Nlevel}  can also be analyzed by means of the linear matrix inequality (LMI) approach combined with the Lyapunov method~\cite{mondie2005exponential}, and the convergence of states can be evaluated in terms of the exponential stability for time-invariant systems~\cite{mondie2005exponential,Newexponential,bliman2002lyapunov,gu2001further,Complexspace,Complexsong2016stability,Switchstability,mahmoud2005new} and time-varying systems~\cite{TimeVarykim2001delay,TimeVaryli1997delay,TimeVaryhien2009exponential}. For example, as studied in \cite{QuantumStochasticfeedback} on a quantum system with delayed feedback, the control performance can be evaluated in terms of the quantum state's convergence to the target state, and the feedback control stability can be investigated with the Lyapunov method~\cite{kuang2018stability}. The above methods based on the traditional stability theory can be potentially generalized to the analysis of the performance of coherent feedback loops with photons, which is the case in our study.

In this paper, we investigate the coherent feedback dynamics established by the interactions among a multi-level atom, cavity and waveguide, as shown in Fig.~\ref{fig:Nlevel}. We propose that the above quantum coherent feedback dynamics can be modeled as a linear control system with a round-trip delay. The linear dynamics can be time-varying or time-invariant, depending on whether there are detunings between the atom and cavity or not.  We show that the evolution of eigenstates of the quantum system can be studied in terms of the exponential stability of a linear system with delay. Thus we can analyze the dynamics of the quantum system from the perspective of linear control theory with the Lyapunov techniques, for both time-invariant and time-varying dynamics according to the existence of the detunning between the atom and cavity. The above linear quantum coherent network with delay can be regarded as a multi-photon source, where the generation of multi-photon states can be tuned by the length of the coherent feedback loop. We also generalize the architecture in Fig.~\ref{fig:Nlevel} with one waveguide feedback loop to the circumstance in Fig.~\ref{fig:multiwaveguide} with parallel waveguides, and study how the distribution of the photonic states among waveguides is influenced by the feedback loop length and the coupling parameters.

The rest of the paper is organized as follows. Section~\ref{Sec:Model} concentrates on the coherent feedback interaction between one waveguide and a cavity coupled with a multi-level atom which is initially excited, especially on the distributions of photons influenced by the feedback design and the exponential stability of this coherent feedback network. In Section~\ref{Sec:detuning}, the feedback system is time varying as the detunings between the cavity and multi-level atom are considered. In Section~\ref{Sec:Parallel}, we generalize to the circumstance that the cavity-QED system is coupled to an array of parallel waveguides. Section~\ref{Sec:Conclusion} concludes this paper.

The reduced Planck constant $\hbar$ is set to be 1 in this paper.
\section{Coherent feedback control of an $N$-level atom with one waveguide}\label{Sec:Model}
\begin{figure}[h]
\centerline{\includegraphics[width=1\columnwidth]{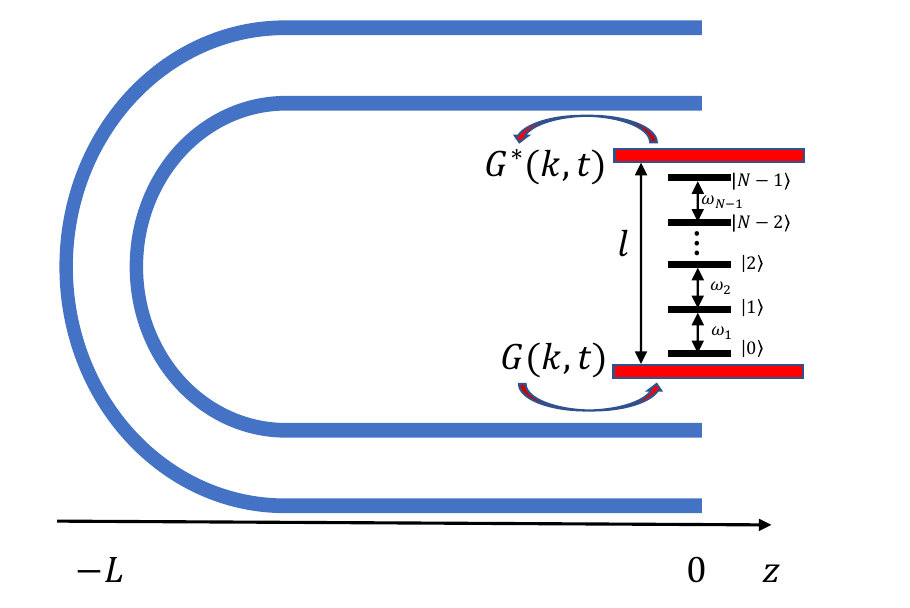}}
\caption{Schematic of a ladder-type $N$-level atom inside a cavity (designated by the two red bars) coupled with a waveguide (designated by the blue band).}
	\label{fig:Nlevel}
\end{figure}
As illustrated in Fig.~\ref{fig:Nlevel}, a ladder-type (also called $\Xi$-type) $N$-level atom is coupled with a cavity and the cavity is coupled to a waveguide of length $2L$. The initially excited atom can emit a photon into the cavity, then the  photon  can be transmitted from the cavity to one side (upper or lower) of the waveguide, and enter back into the cavity from the other side after traveling through the waveguide. Thus, the waveguide with propagating photons  constructs a feedback loop for the cavity-QED system, which can be used to engineer the quantum states  of the system. The length of the cavity is $l$, and it is assumed that $l\ll L$. Physically, this means the cavity can interact with the waveguide with continuous modes~\cite{lang1973laser,gea1990treatment,ding2022quantum}. The overall Hamiltonian of this quantum coherent feedback system can be written as
\begin{small}
\begin{equation} \label{con:Hamiltonian}
\begin{aligned}
H = H_A  + \omega_c a^{\dag}a +\int\omega_k d_k^{\dag}d_k \mathrm{d}k+ H_{I}.
\end{aligned}
\end{equation}
\end{small}%
Here, $H_A = \sum_{n=0}^{N-1}  \tilde{\omega}_n |n\rangle \langle n|$ is the Hamiltonian of the $N$-level atom, whose $n$-th level, represented by $|n\rangle$ with $n=0,1,\cdots,N-1$ as shown in Fig.~\ref{fig:Nlevel},  has the energy $ \tilde{\omega}_n$. The energy gap between two neighboring energy levels of the atom is denoted by $\omega_n= \tilde{\omega}_{n} -  \tilde{\omega}_{n-1}$  ($n=1,2,\cdots,N-1$) as shown in Fig.~\ref{fig:Nlevel}, see also  \cite{PRADynamicsII,pan2017scattering}.
The cavity Hamiltonian is $\omega_c a^{\dag}a$, where $\omega_c$ is the resonant frequency of the cavity and $a$($a^{\dag}$) are the annihilation(creation) operators of the cavity mode. The third item on the RHS of Eq. \eqref{con:Hamiltonian} represents the Hamiltonian of the waveguide of a  continuum of modes, where $\omega_k = kc$ is the frequency of  the continuous mode $k$ with $c$ being the velocity of the field, and $d_k$($d_k^{\dag}$) are the annihilation(creation) operators of the  mode $k$. 
The interaction Hamiltonian $H_I$ on the RHS of Eq. \eqref{con:Hamiltonian}  includes the interaction between the waveguide and the cavity as well as  that between the cavity and the $N$-level atom. As we will investigate the system dynamics in the interaction picture, the specific form of $H_I$ need not be given; instead, its counterpart in the interaction picture, $\tilde{H}_I$, is given in Eq.~(\ref{con:Ham}) below.

If $\omega_1 = \omega_2 =\dots = \omega_{N-1} \equiv  \omega_a=\omega_c$ (the resonant case), then in the interaction picture  $H_I$ can be concisely written  as~\cite{photonfeedback,ding2022quantum,nemet2019stabilizing,whalen2015open,nemet2019time,crowder2020quantum,crowderPRAtraje,regidor2021modeling,AntonPRL}
\begin{small}
\begin{equation} \label{con:Ham}
\begin{aligned}
\tilde{H}_I = &-\sum_{n=1}^{N-1} \gamma_n \left (\sigma_n^-a^{\dag} + \sigma_n^+a\right ) \\
&-  \int  \left [G(k,t) a^{\dag}d_k + G^*(k,t)ad^{\dag}_k\right ]\mathrm{d}k.
\end{aligned}
\end{equation}
\end{small}%
Here, $\sigma_n^- = |n-1\rangle \langle n|$ are the lowering operators that lower the atom's energy  from the $n$th level to the $(n-1)$th energy level, $\sigma_n^+ = |n\rangle \langle n-1|$  are  raising operators that raise the atom's energy from the $(n-1)$th level to the $n$th energy level, $\gamma_n$ is the coupling strength between the cavity and the atomic transition between energy states $|n\rangle$ and $|n-1\rangle$. Thus, $\sigma_n^-a^{\dag}$ means that one photon is created in the cavity via the atom's decaying to a lower energy level, and $\sigma_n^+a$ gives the reverse process, i.e., the atom is excited to a higher energy level by absorbing one photon from the cavity. Analogously,  $a^{\dag}d_k$ and $ad^{\dag}_k$ represent the exchange of a photon between the cavity and  the mode $k$ of the waveguide. The rate of this process is determined by the coupling strength $G(k,t) = G_0\sin(kL)e^{-i(\omega_k-\Delta_0)t}$ between the cavity and   the mode $k$ of the waveguide, where $G_0$ is the amplitude, $\Delta_0 = \omega_c$ is the central mode of the emitted photon, and $\sin(kL)$ represents the superposition of two reciprocal propagating waves in the waveguide emitted from two mirrors of the cavity~\cite{photonfeedback,ding2022quantum}. The continuous modes $k$ in the waveguide are integrated over $(-\infty,+\infty)$ in this paper. More physical interpretations can be found in ~\cite{walls2008quantum,gardiner2004quantum,photonfeedback,fan2010input} and references therein.

Assume that initially the atom is at its highest energy level $N-1$, and the cavity and waveguide are both empty. Thus there are at most $N-1$ photons in the whole network. Moreover, as there are no coherent drives, the number of excitations is preserved during system evolution. The evolution of the quantum state  $|\Psi(t)\rangle$ of the coherent feedback network in  Fig.~\ref{fig:Nlevel} is governed by the Schr\"{o}dinger equation
\begin{small}
\begin{equation} \label{con:SchodingerEQ}
\begin{aligned}
\frac{\mathrm{d}}{\mathrm{d} t}|\Psi(t)\rangle = -i \mathbf{H} |\Psi(t)\rangle,
\end{aligned}
\end{equation}
\end{small}%
where $\mathbf{H}$ is determined according to the representation picture we choose. In the Schr\"{o}dinger picture, $\mathbf{H}$ is $H$ in Eq.~(\ref{con:Hamiltonian}), while in the interaction picture, $\mathbf{H}$ should be $\tilde{H}_I$ in Eq.~(\ref{con:Ham}). For more physical details please refer to Ref.~\cite{photonfeedback}. In the following, we analyze the quantum coherent feedback control dynamics in the interaction picture with the number of excitations being preserved, in this case the quantum state $\ket{\Psi(t)}$ is of the following form~\cite{photonfeedback}
\begin{small}
\begin{equation} \label{con:state}
\begin{aligned}
|\Psi(t)\rangle = &c_0(t)|N-1,0,0\rangle \\
&+ \sum_{j = 1}^{N-1} \sum_{m=0}^{j}\int \cdots \int c_{j,k}^m(t,k_1,\cdots,k_{j-m}) \\
&|N-1-j,m,j-m\rangle \mathrm{d}k_1 \cdots \mathrm{d}k_{j-m},
\end{aligned}
\end{equation}
\end{small}%
where the basis state $|N-1-j,m,j-m\rangle$ means that the atom is at the $(N-1-j)$th level, there are $m$ photons in the cavity and $j-m$ photons in the waveguide,
$c_{j,k}^m(t,k_1,\cdots,k_{j-m})$ is the time-varying amplitude of this state. In particular, $|N-1,0,0\rangle$ with the amplitude $c_0(t)$ means that the atom is at the highest energy level and both the cavity and waveguide are empty. 
The amplitudes of the quantum states are constrained by the normalization
\begin{footnotesize}
\[
|c_0(t)|^2 + \sum_{j=1}^{N-1} \sum_{m=0}^{j}\int \cdots \int |c_{j,k}^m(t,k_1,\cdots,k_{j-m})|^2  \mathrm{d}k_1 \cdots \mathrm{d}k_{j-m} =1,
\]
\end{footnotesize}%
under the initial condition that $c_0(0) = 1$ and all the other amplitudes are zero.

Obviously, $m\leq j \leq N-1$ in Eq. \eqref{con:state}. In what follows we look at three special cases of $|N-1-j,m,j-m\rangle$ for demonstration. Firstly, for $j=m=0$, the initial state $|N-1,0,0\rangle$ means that the atom is at its highest energy level $\ket{N-1}$ and there are no photons in either the cavity or waveguide. Secondly,  when $m=j$,  the waveguide is empty.  Accordingly, $c_{j,k}^m(t,k_1,\cdots,k_{j-m})$ is only a function of time $t$, which can be rewritten as $c_{j,k}^j(t)$.  Finally, when $m < j$, for the example $j-m=2$, the state $|N-1-j,j-2,2\rangle$ means that the atom is at the $(N-1-j)$th level, there are $j-2$ photons in the cavity, and the other two photons are in the waveguide. Moreover, as the waveguide contains a continuum of modes, the two photons in the waveguide can be in any two of these modes. Thus,  in this case, the multiple integral in Eq. \eqref{con:state}  is $\int \int c_{j,k}^{j-2}(t,k_1,k_2) |N-1-j,j-2,2\rangle\mathrm{d}k_1  \mathrm{d}k_2$.

\newtheorem{assumption}{Assumption}

The evolution of quantum state amplitudes can be derived by substituting Eq.~(\ref{con:state}) into the Schr\"{o}dinger equation (\ref{con:SchodingerEQ}). By comparing the coefficients of each basis state $|N-1-j,m,j-m\rangle$ on both sides, we can yield the following set of equations for the amplitudes. For examples, the operator $\sigma_n^- a^{\dag} $ in the interaction Hamiltonian \eqref{con:Ham} acting on the state $|n,m,j-m\rangle$  creates the state $|n-1,m+1,j-m\rangle$, the operator $a^{\dag}d_k$ applied upon the state $|n,m,j-m\rangle $ creates the state $|n,m+1,j-m-1\rangle$, and similarly for other conjugate components in Eq.~(\ref{con:state}). In summary, we have the ordinary differential equations (ODEs) for the amplitudes as 
\begin{footnotesize}
\begin{subequations} \label{con:controleq}
\begin{align}
&\dot{c}_0(t) =  i\gamma_{N-1}  c_{1,k}^{1}(t),   \label{Nmodel1}\\
&\dot{c}_{j,k}^m(t,k_1,\cdots,k_{j-m}) = i\sqrt{m}\gamma_{N-j}  c_{j-1,k}^{m-1}(t,k_1,\cdots,k_{j-m})\notag\\
& + i\sqrt{m+1}\gamma_{N-j-1}  c_{j+1,k}^{m+1}(t,k_1,\cdots,k_{j-m}) \notag\\
&+i\sum_{p=1}^{j-m+1}\int G(k_p,t) c_{j,k}^{m-1}(t,k_1,\cdots,k_{p-1},k_p,k_{p+1}\dots,k_{j-m+1}) \mathrm{d}k_p \notag\\
&+i\sum_{p=1}^{j-m} G^*(k_p,t) c_{j,k}^{m+1}(t,k_1,\cdots,k_{p-1},k_{p+1},\dots, k_{j-m}), m>0, \label{Nmodel2} \\
&\dot{c}_{j,k}^0(t,k_1,\cdots,k_{j}) =  i\gamma_{N-j-1}  c_{j+1,k}^{1}(t,k_1,\cdots,k_{j}) \notag\\
&+i\sum_{p=1}^{j} G^*(k_p,t) c_{j,k}^{1}(t,k_1,\cdots,k_{p-1},k_{p+1},\dots, k_{j}), m=0. \label{Nmodel3}
\end{align}
\end{subequations}
\end{footnotesize}%

Eq.~(\ref{Nmodel1}) means that the atom at the energy level $|N-1\rangle$ can emit one photon into the cavity by decaying to the energy level $|N-2\rangle$, and the rate of this process is determined by the coupling rate $\gamma_{N-1}$. In this case,  the overall system realizes the transition from $|N-1,0,0\rangle$ to $|N-2,1,0\rangle$. 
Next we look at Eq. \eqref{Nmodel2}. A photon in the cavity may be absorbed by the atom (described by the second item on the RHS of Eq. \eqref{Nmodel2}) or enter into the waveguide (described by the fourth item on the RHS of Eq. \eqref{Nmodel2}). Besides, the atom may emit a photon into the cavity  (described by the first item on the RHS of Eq. \eqref{Nmodel2}), and the photon can also be transmitted from the waveguide to the cavity (described by the third item on the RHS of Eq. \eqref{Nmodel2}).
Finally, we look at Eq. \eqref{Nmodel3}. The coefficient $c_{j,k}^0(t,k_1,\cdots,k_{j})$ is for the state  $|N-1-j,0,j\rangle$, which means an empty cavity. This case is resulted from the fact that the only photon in the cavity is absorbed by the atom (described by the first item on the RHS of Eq. \eqref{Nmodel3}) or transmitted from the cavity to the waveguide (described by the second item on the RHS of Eq. \eqref{Nmodel3}.  

In the special case of $N=2$, there is only one excitation in the system, thus there can be at most one photon in the cavity or waveguide, and the system in Eq.~\eqref{con:controleq} reduces to Eqs.~(10-12) in Ref.~\cite{photonfeedback}.

According to the mathematical calculations in Appendix~\ref{Sec:ModeldelayAppend}, Eq.~(\ref{con:controleq}) can be rewritten as a system of time-delayed ODEs:
\begin{footnotesize}
\begin{subequations} \label{con:controleq2}
\begin{align}
&\dot{c}_0(t) =  i\gamma_{N-1}  c_{1k}^{1}(t),   \label{delaymodel1}\\
&\dot{c}_{j,k}^m(t,k_1,\cdots,k_{j-m}) = i\sqrt{m}\gamma_{N-j}  c_{j-1,k}^{m-1}(t,k_1,\cdots,k_{j-m}) \notag\\
&+ i\sqrt{m+1}\gamma_{N-j-1}  c_{j+1,k}^{m+1}(t,k_1,\cdots,k_{j-m}) \notag\\
&-\frac{G_0^2}{4c} \sum_{p=1}^{j-m+1} [c_{j,k}^{m}(t,k_1,\cdots,k_{p-1},k_{p+1},\dots, k_{j-m}) \notag\\
&- e^{i\Delta_0\tau} c_{j,k}^{m}(t-\tau,k_1,\cdots,k_{p-1},k_{p+1},\dots, k_{j-m}) ] \notag\\
&+i\sum_{p=1}^{j-m} G^*(k_p,t) c_{j,k}^{m+1}(t,k_1,\cdots,k_{p-1},k_{p+1},\dots, k_{j-m-1}), m>0, \label{delaymodel2}\\
&\dot{c}_{j,k}^0(t,k_1,\cdots,k_{j}) =  i\gamma_{N-j-1}  c_{j+1,k}^{1}(t,k_1,\cdots,k_{j}) \notag\\
&+i\sum_{p=1}^{j} G^*(k_p,t) c_{j,k}^{1}(t,k_1,\cdots,k_{p-1},k_{p+1},\dots, k_{j}), m=0, \label{delaymodel3}
\end{align}
\end{subequations}
\end{footnotesize}%
where $\tau = 2L/c$ denotes the round-trip delay of the transmission of photons in the waveguide.
Compared with Eq.~(\ref{con:controleq}), the third term on the right-hand side of Eq.~(\ref{delaymodel2}) shows that the amplitude $c_{j,k}^m(t,k_1,\cdots,k_{j-m})$ is not only influenced by its present value, but also the  past value due to the round-trip delay $\tau$. 
This indicates how the feedback loop mediated by the waveguide influences the dynamics of the atom.

In summary, we have derived the control equation of the coherent feedback network in Eq.~(\ref{con:controleq}), or equivalently Eq.~(\ref{con:controleq2}). Based on this general control equation, in the following subsection we take the three-level atom as an example  to study two physically relevant scenarios: one is that the waveguide is so short compared with the lifetime of atom that  the delay is negligible (part  II.A.1), and the other is that the delay is longer than the lifetime of the atom (part  II.A.2).

\subsection{Example: feedback control of the ladder-type three-level atom}
In this section we take the three-level atom as an example to demonstrate the system modeling in the previous section. In this case, $N=3$. Assume  the initial system state is $|\Psi(0)\rangle = |2,0,0\rangle$. In Eq.~(\ref{con:controleq2}), $c_{1,k}^m$ with $m=0,1$ indicates that there is one photon either in the cavity or the waveguide, and $c_{2,k}^m$ with $m=0,1,2$ indicates that there are overall two photons in the cavity and waveguide.  
According to Eq.~(\ref{con:state}), the quantum state is
\begin{small}
\begin{equation} \label{con:statethreelevel}
\begin{aligned}
|\Psi(t)\rangle = &c_0(t)|2,0,0\rangle + c_{1,k}^1(t) |1,1,0\rangle + \int c_{1,k}^0(t,k) |1,0,1\rangle \mathrm{d}k \\
& + c_{2,k}^2(t) |0,2,0\rangle + \int c_{2,k}^1(t,k) |0,1,1\rangle \mathrm{d}k \\
& + \int \int c_{2,k}^0(t,k_1,k_2) |0,0,2\rangle  \mathrm{d}k_1  \mathrm{d}k_2.
\end{aligned}
\end{equation}
\end{small}%
The interaction Hamiltonian in Eq.~(\ref{con:Ham}) reads
\begin{small}
\begin{equation} \label{con:Ham3level}
\begin{aligned}
\tilde{H}_I = &- \gamma_1 (\sigma_1^-a^{\dag} + \sigma_1^+a) - \gamma_2 (\sigma_2^-a^{\dag} + \sigma_2^+a) \\
&- \int  \left[G(k,t) a^{\dag}d_k + G^*(k,t)ad^{\dag}_k\right]\mathrm{d}k.
\end{aligned}
\end{equation}
\end{small}%
By Eqs.~(\ref{con:controleq},\ref{con:controleq2}), we have
\begin{footnotesize}
\begin{subequations} \label{con:controleqthreelevelDelay}
\begin{align}
&\dot{c}_0(t) =  i\gamma_{2}  c_{1,k}^{1}(t),   \label{threedelay1}\\
&\dot{c}_{1,k}^1(t) =  i\gamma_2 c_0(t) + i \gamma_1 c_{2,k}^2(t) -\frac{\left|G_0\right|^2}{4c} \left [c_{1,k}^1(t) -e^{i\Delta_0\tau} c_{1,k}^1(t-\tau)\right], \label{threedelay2}\\
&\dot{c}_{2,k}^2(t) = i\gamma_1 c_{1,k}^1(t)  -  \frac{\left|G_0\right|^2}{4c} \left [ c_{2,k}^2(t) -e^{i\Delta_0\tau} c_{2,k}^2(t-\tau)\right ], \label{threedelay4}\\
&\dot{c}_{1,k}^0(t,k) =i G^*(k,t) c_{1,k}^1(t)  + i\gamma_1 c_{2,k}^1(t,k), \label{threedelay3}\\
&\dot{c}_{2,k}^1(t,k)  = i\gamma_1 c_{1,k}^0(t,k) + i G^*(k,t) c_{2,k}^2(t) \notag\\
&~~~~~~~~~~~~~~ -\frac{\left|G_0\right|^2 }{2c} \left [c_{2,k}^1(t,k) - e^{i\Delta_0\tau}c_{2,k}^1(t-\tau,k) \right ],\label{threedelay5}\\
&\dot{c}_{2,k}^0(t,k_1,k_2) = iG^*(k_2,t) c_{2,k}^1(t,k_1) +  iG^*(k_1,t) c_{2,k}^1(t,k_2). \label{threedelay6}
\end{align}
\end{subequations}
\end{footnotesize}%
Denote $\kappa = \left|G_0\right|^2/4c$ and apply the Laplace transform to Eqs.~(\ref{threedelay1}-\ref{threedelay4}),  we get
\begin{footnotesize}
\begin{subequations} \label{con:Laplace1to3}
\begin{numcases}{}
sC_0(s) -1  =  i\gamma_{2}  C_{1,k}^{1}(s),   \label{Lapthreedelay100}\\
sC_{1,k}^1(s) =  i\gamma_2 C_0(s) + i \gamma_1 C_{2,k}^2(s) \notag\\
~~~~~~~~~~~~~ -\kappa \left [C_{1,k}^1(s) -e^{i\Delta_0\tau} e^{-s\tau}C_{1,k}^1(s)\right ], \label{Lapthreedelay200}\\
sC_{2,k}^2(s) = i\gamma_1 C_{1,k}^1(s)  -  \kappa \left [ C_{2,k}^2(s) -e^{i\Delta_0\tau} e^{-s\tau} C_{2,k}^2(s) \right], \label{Lapthreedelay40}
\end{numcases}
\end{subequations}
\end{footnotesize}%
where $C_0(s)$, $C_{j,k}^j(s)$ are the frequency counterparts of the time-domain functions $c_0(t)$, $c_{j,k}^j(t)$, respectively. Solving Eq.~\eqref{con:Laplace1to3} we have
\begin{scriptsize}
\begin{subequations} \label{con:Laplace1to32}
\begin{numcases}{}
C_0(s)   =  \frac{ i\gamma_{2}  C_{1,k}^{1}(s) + 1}{s},   \label{Lapthree0delay1}\\
sC_{1,k}^1(s) =   \frac{ -\gamma_{2}^2  C_{1,k}^{1}(s) + i\gamma_2}{s} - \frac{\gamma_1^2}{s+ \kappa (1- e^{(i\Delta_0-s)\tau})} C_{1,k}^1(s) \notag\\
~~~~~~~~~~~~~~~~ -\kappa \left (1- e^{(i\Delta_0-s)\tau}\right )C_{1,k}^1(s) , \label{Lapthree0delay2}\\
C_{2,k}^2(s) = \frac{i\gamma_1}{s+ \kappa (1- e^{(i\Delta_0-s)\tau})} C_{1,k}^1(s). \label{Lapthree0delay4}
\end{numcases}
\end{subequations}
\end{scriptsize}%

Finally, Eq.~(\ref{Lapthree0delay2}) can be rewritten as
\begin{scriptsize}
\begin{equation} \label{con:C1k1s0}
\begin{aligned}
~\left [s^2+ \gamma_2^2 + \frac{s\gamma_1^2}{s+ \kappa \left (1- e^{(i\Delta_0-s)\tau}\right )} +   \kappa (1- e^{(i\Delta_0-s)\tau})s \right ]C_{1,k}^1(s) = i\gamma_2.
\end{aligned}
\end{equation}
\end{scriptsize}%

In what follows we discuss two scenarios depending on the relationship between the lifetime of an atom and the transmission delay in the waveguide.

\subsubsection{Feedback delay is much smaller than atom lifetime}
\

When the length of the waveguide $2L$ is small and the delay $\tau  \ll 1$ is much smaller than the lifetime of the atom, then $e^{-is\tau} \approx 1$ when $s$ is close to the origin of the complex plane.

We denote the amplitudes in Eq.~(\ref{con:statethreelevel}) in this parameter setting, i.e., $c_0(t)$, $c_{j,k}^m(t)$ as  $\tilde{c}_0(t)$, $\tilde{c}_{j,k}^m(t)$, and the Laplace transformation in Eq.~(\ref{con:Laplace1to3}), i.e., $C_0(s)$, $C_{j,k}^m(s)$ as $\tilde{C}_0(s)$, $\tilde{C}_{j,k}^m(s)$, respectively.
Then Eq.~(\ref{con:C1k1s0}) can be simplified as:
\begin{scriptsize}
\begin{equation} \label{con:C1k1s000}
\begin{aligned}
&\left [(s^2+ \gamma_2^2) (s+ \kappa (1- e^{i\Delta_0\tau})) + s\gamma_1^2 \right.\\
&\left.+   \kappa (1- e^{i\Delta_0\tau})s (s+ \kappa (1- e^{i\Delta_0\tau})) \right ]\tilde{C}_{1,k}^1(s) = i\gamma_2 (s+ \kappa (1- e^{i\Delta_0\tau})).
\end{aligned}
\end{equation}
\end{scriptsize}%
Consequently,
\begin{scriptsize}
\begin{equation} \label{con:C1k1s}
\begin{aligned}
&~~~~\tilde{C}_{1,k}^1(s)\\
&= \frac{i\gamma_2 (s+ \kappa (1- e^{i\Delta_0\tau}))}{[(s^2+ \gamma_2^2) (s+ \kappa (1- e^{i\Delta_0\tau})) + s\gamma_1^2 +  \kappa (1- e^{i\Delta_0\tau})s (s+ \kappa (1- e^{i\Delta_0\tau})) ]}\\
&=\frac{i\gamma_2 (s+ \kappa (1- e^{i\Delta_0\tau}))}{s^3 +2\kappa (1- e^{i\Delta_0\tau})s^2 + (\gamma_1^2 +\gamma_2^2 + \kappa^2(1- e^{i\Delta_0\tau})^2)s + \gamma_2^2 \kappa (1- e^{i\Delta_0\tau})}.
\end{aligned}
\end{equation}
\end{scriptsize}%

Using the Sign Pair Criterion (SPC) in Refs.~\cite{dongCui2023dynamics,sivanandam2012algebraic}, we can get the generalized Hurwitz matrix of the denominator of Eq.~(\ref{con:C1k1s}) and prove that all its eigenvalues are in the LHS of the complex plane. The proof is omitted because of page limitations.

Denote $\tilde{c}_0(\infty)=\lim_{t\rightarrow \infty}\tilde{c}_{0}(t)$ and $\tilde{c}_{j,k}^m(\infty) = \lim_{t\rightarrow \infty}\tilde{c}_{j,k}^m(t)$ as the steady values of the amplitudes. We have the following results.

\newtheorem{mypro}{Proposition}
\begin{mypro} \label{photonstate}
Let $\tau \ll 1$. When $\Delta_0\tau \neq 2n\pi$, eventually there are two photons in the waveguide. When $\Delta_0\tau = 2n\pi$, the atom oscillates in the cavity and there are no photons in the waveguide.
\end{mypro}
\begin{proof}
When $\Delta_0\tau \neq 2n\pi$, by Eq.~\eqref{con:C1k1s} and the final value theorem, we have $\tilde{c}_{1,k}^1(\infty) = \lim_{s\rightarrow 0}s\tilde{C}_{1,k}^1(s) = 0$. Moreover, as $1- e^{(i\Delta_0-s)\tau} \neq 0$, by Eqs.~\eqref{Lapthree0delay1} and \eqref{Lapthreedelay40}, we get
\begin{small}
\begin{equation} \label{con:C0inf1}
\begin{aligned}
\tilde{c}_0(\infty) &= \lim_{s\rightarrow 0}s\tilde{C}_0(s) = \lim_{s\rightarrow 0} [ i\gamma_{2}  \tilde{C}_{1,k}^{1}(s) + 1]\\
&= \frac{-\gamma_2^2  \kappa (1- e^{i\Delta_0\tau}) }{\gamma_2^2 \kappa (1- e^{i\Delta_0\tau})} + 1  =0,
\end{aligned}
\end{equation}
\end{small}%
and
\begin{small}
\begin{equation} \label{con:C2k2inf1}
\begin{aligned}
\tilde{c}_{2,k}^2(\infty) &= \lim_{s\rightarrow 0}s\tilde{C}_{2,k}^2(s) \\
&= \lim_{s\rightarrow 0}  s \frac{i\gamma_1}{s+ \frac{G_0^2}{4c} \left (1- e^{(i\Delta_0-s)\tau}\right )} \tilde{C}_{1,k}^1(s)  =0.
\end{aligned}
\end{equation}
\end{small}%
Thus, the waveguide eventually contains two photons.

On the other hand, when $\Delta_0\tau = 2n\pi$, where $n=0,1,2,\cdots$,  $1- e^{i\Delta_0\tau} = 0$.
Then
\begin{small}
\begin{equation} \label{con:C1k1sZero}
\begin{aligned}
\tilde{C}_{1,k}^1(s) &=\frac{i\gamma_2 s}{s^3  + (\gamma_1^2 +\gamma_2^2 )s } =\frac{i\gamma_2 }{s^2  + (\gamma_1^2 +\gamma_2^2 ) }.
\end{aligned}
\end{equation}
\end{small}%
Applying the inverse Laplace transform to  $\tilde{C}_{1,k}^1(s)$ in Eq.~\eqref{con:C1k1sZero} yields
\[
\tilde{c}_{1,k}^1(t) = \frac{i\gamma_2}{\sqrt{\gamma_1^2 +\gamma_2^2}}\sin\left(\sqrt{\gamma_1^2 +\gamma_2^2} t\right).
\]
Moreover, by Eqs.~\eqref{Lapthree0delay1} and \eqref{con:C1k1sZero},
\begin{small}
\begin{equation} \label{con:C0sZero}
\begin{aligned}
\tilde{C}_0(s) &=\frac{ i\gamma_{2}  \tilde{C}_{1,k}^{1}(s) + 1}{s} \\
&=\frac{\gamma_1^2}{\gamma_1^2 + \gamma_2^2} \frac{1}{s} + \frac{\gamma_2^2}{\gamma_1^2 + \gamma_2^2}\frac{s}{s^2 + \gamma_1^2 + \gamma_2^2}.
\end{aligned}
\end{equation}
\end{small}%
Therefore, applying the inverse Laplace transform we get
\begin{small}
\begin{equation} \label{con:C0Zerot}
\begin{aligned}
\tilde{c}_{0}(t) &=\frac{\gamma_1^2}{\gamma_1^2 + \gamma_2^2}  \Theta(t) + \frac{\gamma_2^2}{\gamma_1^2 + \gamma_2^2} \cos\left(\sqrt{\gamma_1^2 + \gamma_2^2} t\right),
\end{aligned}
\end{equation}
\end{small}%
where $ \Theta(t)$ represents the Heaviside step function. Finally, by Eqs.~\eqref{Lapthree0delay4} and \eqref{con:C1k1sZero},
\begin{small}
\begin{equation} \label{con:C2k2sZero}
\begin{aligned}
\tilde{C}_{2,k}^2(s) &=\frac{i\gamma_1}{s+ \frac{G_0^2}{4c} (1- e^{(i\Delta_0-s)\tau})} \tilde{C}_{1,k}^1(s)\\
&=\frac{-\gamma_1\gamma_2}{\gamma_1^2 +\gamma_2^2 }\frac{1}{s} + \frac{\gamma_1\gamma_2}{\gamma_1^2 +\gamma_2^2 }\frac{s}{s^2  + (\gamma_1^2 +\gamma_2^2 )}.\\
\end{aligned}
\end{equation}
\end{small}%
Applying the inverse Laplace transform to $C_{2,k}^2(s)$ yields
\begin{small}
\begin{equation} \label{con:C2k2sZerot}
\begin{aligned}
\tilde{c}_{2,k}^2(t) &=\frac{-\gamma_1\gamma_2}{\gamma_1^2 +\gamma_2^2 }\Theta(t) + \frac{\gamma_1\gamma_2}{\gamma_1^2 +\gamma_2^2 }\cos\left (\sqrt{\gamma_1^2 +\gamma_2^2}t \right)\\
&=\frac{\gamma_1\gamma_2}{\gamma_1^2 +\gamma_2^2 }\left [\cos\left (\sqrt{\gamma_1^2 +\gamma_2^2}t \right)-1\right].
\end{aligned}
\end{equation}
\end{small}%
Notice that
\begin{small}
\begin{equation} \label{con:sumpop}
\begin{aligned}
&~~~~|\tilde{c}_0(t)|^2+ |\tilde{c}_{1,k}^1(t)|^2 +|\tilde{c}_{2,k}^2(t)|^2 \\
&=\frac{\gamma_1^4}{\left ( \gamma_1^2 +\gamma_2^2\right )^2} + \frac{\gamma_1^2\gamma_2^2}{\left ( \gamma_1^2 +\gamma_2^2\right )^2} + \frac{\gamma_2^2}{\gamma_1^2 +\gamma_2^2} \\
&+ \left [\frac{\gamma_2^4}{\left ( \gamma_1^2 +\gamma_2^2\right )^2} + \frac{\gamma_1^2\gamma_2^2}{\left ( \gamma_1^2 +\gamma_2^2\right )^2} - \frac{\gamma_2^2}{\gamma_1^2 +\gamma_2^2}\right ]\cos^2\left (\sqrt{\gamma_1^2 +\gamma_2^2}t \right) \\
&=1.
\end{aligned}
\end{equation}
\end{small}%
The normalization condition of populations yields
\[
\tilde{c}_{2,k}^1(t,k)= \tilde{c}_{2,k}^0(t,k) = \tilde{c}_{2,k}^2(t,k_1,k_2) = 0.
\]
 In words, the atom oscillates in the cavity while there are no photons in the waveguide.
\end{proof}

Specially, when $\gamma_1= \gamma_2$, the amplitude of the high-frequency oscillating component in $\left|\tilde{c}_{1,k}^1(t)\right|^2$ is twice of that in $\left|\tilde{c}_0(t)\right|^2$ and $\left|\tilde{c}_{2,k}^2(t)\right|^2$, implying that the atom can be at the middle energy level $|1,1,0\rangle$ by emitting one photon from $|2,0,0\rangle$, or absorbing one photon from $|0,2,0\rangle$.

\newtheorem{remark}{Remark}

In summary, when $\tau \ll 1$, the amplitudes can be approximately and analytically solved 
as in Eqs.~(\ref{con:C1k1s000} - \ref{con:C2k2sZerot}). 
To further clarify and compare the differences between
the dynamics of the original Eq.~(\ref{con:controleq}) and the simplified format in Eqs.~(\ref{con:C1k1s000} - \ref{con:C2k2sZerot}), we do the simulation based on Eq.~(\ref{con:controleq}), $N=3$, in Fig.~\ref{fig:threelevelfeedback}, and the results agree with the analytical results in Eqs.~(\ref{con:C1k1s000} - \ref{con:C2k2sZerot}). The details are described below.

Take $\Delta_0 =50$, $G_0 = 0.2$, $\gamma_1 = \gamma_2 =0.3$.  In Fig.~\ref{fig:threelevelfeedback} the simulation results are compared according to the number of photons in the waveguide. In the upper three subfigures labeled as (1-1), (1-2) and (1-3) respectively, $\Delta_0\tau = 2\pi$. The populations  that there are no photons in the waveguide oscillate, as shown in Fig.~\ref{fig:threelevelfeedback} (1-1). The populations that there are one or two photons in the waveguide finally converge to zero when $t= 200\tau$, as shown in Fig.~\ref{fig:threelevelfeedback} (1-2) and (1-3). Hence, asymptotically there are no photons in the waveguide.  On the other hand, in the lower three subfigures labeled as (2-1), (2-2) and (2-3) respectively, $\Delta_0\tau = 3\pi$.  $|c_0(t)|^2$, $|c_{1,k}^1(t)|^2$ and $|c_{2,k}^2(t)|^2$ converge to zero as shown in Fig.~\ref{fig:threelevelfeedback} (2-1), the single-photon population converges to zero when $t= 200\tau$, as shown in Fig.~\ref{fig:threelevelfeedback} (2-2), and finally there are two photons in the waveguide as shown in Fig.~\ref{fig:threelevelfeedback} (2-3).
In both of the simulations, $\tau$ is small because $\Delta_0 \gg \pi$.

\begin{figure}[h]
\centerline{\includegraphics[width=1\columnwidth]{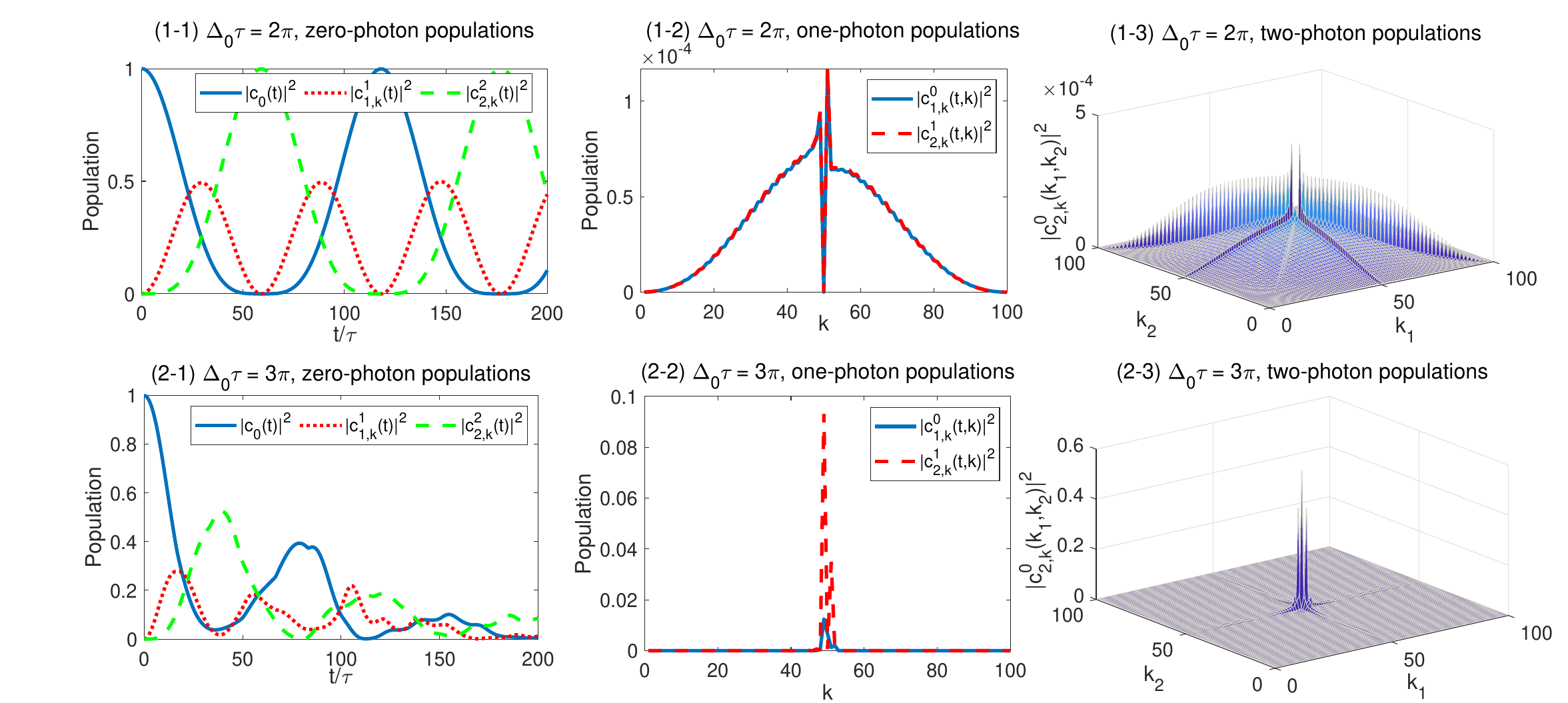}}
\caption{The populations of the zero-, one- and two-photon states in the waveguide.}
	\label{fig:threelevelfeedback}
\end{figure}

\medskip
\subsubsection{Feedback delay is much larger than  atom's lifetime}

When the waveguide is so long that the induced round-trip delay $\tau$ is much larger than the lifetime of the atom, we only need to  consider the dynamics of the $\Xi$-type three-level atom in Eq.~(\ref{con:controleqthreelevelDelay}) when $t< \tau$ and  safely neglect the delay-induced part. In this case, the quantum dynamics in Eq.~(\ref{con:controleqthreelevelDelay}) is equivalent to the following equation by replacing $c_0(t)$, $c_{j,k}^m(t)$ with  $\hat{c}_0(t)$, $\hat{c}_{j,k}^m(t)$, and $C_0(s)$, $C_{j,k}^m(s)$ with $\hat{C}_0(s)$, $\hat{C}_{j,k}^m(s)$, respectively.
\begin{small}
\begin{subequations} \label{con:longwaveguide}
\begin{align}
&\dot{\hat{c}}_0(t) =  i\gamma_{2}  \hat{c}_{1,k}^{1}(t),   \label{longthreedelay1}\\
&\dot{\hat{c}}_{1,k}^1(t) =  i\gamma_2 \hat{c}_0(t) + i \gamma_1 \hat{c}_{2,k}^2(t) -\kappa \hat{c}_{1,k}^1(t) , \label{longthreedelay2}\\
&\dot{\hat{c}}_{2,k}^2(t) = i\gamma_1 \hat{c}_{1,k}^1(t)  -  \kappa \hat{c}_{2,k}^2(t) , \label{longthreedelay4}\\
&\dot{\hat{c}}_{1,k}^0(t,k) =i G^*(k,t) \hat{c}_{1,k}^1(t)  + i\gamma_1 \hat{c}_{2,k}^1(t,k), \label{longthreedelay3}\\
&\dot{\hat{c}}_{2,k}^1(t,k)  = i\gamma_1 \hat{c}_{1,k}^0(t,k) + i G^*(k,t) \hat{c}_{2,k}^2(t),\label{longthreedelay5}\\
&\dot{\hat{c}}_{2,k}^0(t,k_1,k_2) = iG^*(k_2,t) \hat{c}_{2,k}^1(t,k_1) +  iG^*(k_1,t) \hat{c}_{2,k}^1(t,k_2). \label{longthreedelay6}
\end{align}
\end{subequations}
\end{small}%

Notice that Eqs.~(\ref{longthreedelay1}-\ref{longthreedelay4}) are ODEs with constant coefficients. Applying the Laplace transform to $\hat{c}_0(t)$, $\hat{c}_{1,k}^1(t)$ and $\hat{c}_{2,k}^2(t)$ yields
\begin{small}
\begin{subequations} \label{con:LongLaplace1to32}
\begin{align}
&\hat{C}_0(s)   =  \frac{ i\gamma_{2}  \hat{C}_{1,k}^{1}(s) + 1}{s},   \label{Lapthreedelay10000}\\
&s\hat{C}_{1,k}^1(s) =   \frac{ -\gamma_{2}^2  \hat{C}_{1,k}^{1}(s) + i\gamma_2}{s} - \frac{\gamma_1^2}{s+ \kappa} \hat{C}_{1,k}^1(s) -\kappa \hat{C}_{1,k}^1(s) , \label{Lapthreedelay2000}\\
&\hat{C}_{2,k}^2(s) = \frac{i\gamma_1}{s+ \kappa} \hat{C}_{1,k}^1(s),\label{Lapthreedelay4}
\end{align}
\end{subequations}
\end{small}%
respectively, which can be rewritten as
\begin{small}
\begin{equation} \label{con:LongC1k1s}
\begin{aligned}
\hat{C}_{1,k}^1(s)
&=\frac{i\gamma_2 (s+\kappa)}{s^3 + 2\kappa s^2 + s(\gamma_1^2 + \gamma_2^2 + \kappa^2) + \kappa\gamma_2^2},
\end{aligned}
\end{equation}
\end{small}%
\begin{small}
\begin{equation} \label{con:LongC0s}
\begin{aligned}
\hat{C}_0(s) &= \frac{s^2 + 2\kappa s +\gamma_1^2 + \kappa^2}{s^3 + 2\kappa s^2 +s(\gamma_1^2 + \gamma_2^2 +\kappa^2) + \kappa\gamma_2^2},
\end{aligned}
\end{equation}
\end{small}%
and
\begin{small}
\begin{equation} \label{con:LongC2k2s}
\begin{aligned}
\hat{C}_{2,k}^2(s) &=\frac{-\gamma_1\gamma_2 }{s^3 + 2\kappa s^2 + s(\gamma_1^2 + \gamma_2^2 + \kappa^2) + \kappa\gamma_2^2},
\end{aligned}
\end{equation}
\end{small}%
respectively.
Clearly, according to the final value theorem, we get
\[
\lim_{t\rightarrow\infty} \hat{c}_0(t) = \lim_{t\rightarrow\infty} \hat{c}_{1,k}^1(t)=\lim_{t\rightarrow\infty} \hat{c}_{2,k}^2(t) = 0.
\]
That is,  when the feedback loop constructed by the waveguide is infinitely long,  eventually the atom settles to its ground state, the cavity is empty, and there are two photons in the waveguide.

\section{Coherent feedback control with detunings} \label{Sec:detuning}
For a general $N$-level ladder-type atom, the interaction Hamiltonian $\tilde{H}_I$ in Eq.~\eqref{con:Ham} should be generalized to~\cite{detune},
\begin{small}
\begin{equation} \label{con:HamDetune}
\begin{aligned}
\tilde{H}_I = &-\sum_{n=1}^{N-1} \gamma_n \left (e^{-i\delta_nt}\sigma_n^-a^{\dag} + e^{i\delta_nt}\sigma_n^+a\right )\\
&- \int  \left [G(k,t) a^{\dag}d_k + G^*(k,t)ad^{\dag}_k\right ] \mathrm{d}k,
\end{aligned}
\end{equation}
\end{small}%
where $\delta_n = \omega_n-\omega_c$, ($n = 1,2,\cdots,N-1$),  represents the detuning between the $n$-th level of the atom and the cavity. Accordingly,  Eq.~(\ref{con:controleq2}) should be modified to
\begin{small}
\begin{subequations} \label{con:controleq2Detune}
\begin{align}
&\dot{c}_0(t) =  i\gamma_{N-1}  e^{i\delta_{N-1}t} c_{1k}^{1}(t),   \label{Ddelaymodel100}\\
&\dot{c}_{j,k}^m = i\sqrt{m}\gamma_{N-j}  e^{-i\delta_{N-j}t} c_{j-1,k}^{m-1}(t,k_1,\cdots,k_{j-m}) \notag\\
&+ i\sqrt{m+1}\gamma_{N-j-1}  e^{i\delta_{N-j-1}t} c_{j+1,k}^{m+1}(t,k_1,\cdots,k_{j-m}) \notag\\
&+i\sum_{p=1}^{j-m} G^*(k_p,t) c_{j,k}^{m+1}(t,k_1,\cdots,k_{p-1},k_{p+1},\dots, k_{j-m}) \notag\\
&-\kappa \sum_{p=1}^{j-m+1} \left [c_{j,k}^{m}(t,k_1,\cdots,k_{p-1},k_{p+1},\dots, k_{j-m}) \right.\notag\\
&\left.- e^{i\Delta_0\tau} c_{j,k}^{m}(t-\tau,k_1,\cdots,k_{p-1},k_{p+1},\dots, k_{j-m}) \right ], m>0, \label{Ddelaymodel2}\\
&\dot{c}_{j,k}^0(t,k_1,\cdots,k_{j}) =  i\gamma_{N-j-1}   e^{i\delta_{N-j-1}t}  c_{j+1,k}^{1}(t,k_1,\cdots,k_{j}) \notag\\
&+i\sum_{p=1}^{j} G^*(k_p,t) c_{j,k}^{1}(t,k_1,\cdots,k_{p-1},k_{p+1},\dots, k_{j}), m=0. \label{Ddelaymodel300}
\end{align}
\end{subequations}
\end{small}%

When $j=m$, all the emitted photons are in the cavity, then $c_{j,k}^m (t,k_1,\cdots,k_{j-m})$ is only a function of time $t$. Thus we can define a time-domain vector
\begin{small}
\begin{equation} \label{con:statevector}
\begin{aligned}
X(t) = \left[c_0(t), c_{1,k}^1(t),\cdots, c_{N-1,k}^{N-1}(t)\right]^{\top} \in \mathbf{C}^{N},
\end{aligned}
\end{equation}
\end{small}%
where the superscript $\top$ represents the transpose of a vector. Then  from Eq.~(\ref{con:controleq2Detune}) we can get
\begin{small}
\begin{subequations} \label{con:DetuneNlevelDelay}
\begin{numcases}{}
\dot{c}_0(t) =  i\gamma_{N-1}  e^{i\delta_{N-1}t} c_{1k}^{1}(t),   \label{Ddelaymodel1}\\
\dot{c}_{1,k}^1(t) =  i\gamma_{N-1} e^{-i\delta_{N-1}t}c_0(t) + i \sqrt{2}\gamma_{N-2} e^{i\delta_{N-2}t} c_{2,k}^2(t) \notag\\
~~~~~~~~~~~ -\kappa \left [c_{1,k}^1(t) -e^{i\Delta_0\tau} c_{1,k}^1(t-\tau) \right ], \label{Dthreedelay2}\\
\dot{c}_{j,k}^j(t) =
i\sqrt{j}\gamma_{N-j}  e^{-i\delta_{N-j}t} c_{j-1,k}^{j-1}(t)\notag\\
~~~~~~~~~~~+ i\sqrt{j+1}\gamma_{N-j-1}  e^{i\delta_{N-j-1}t} c_{j+1,k}^{j+1}(t) \notag\\
~~~~~~~~~~~-\kappa  \left [c_{j,k}^{j}(t) - e^{i\Delta_0\tau} c_{j,k}^{j}(t-\tau) \right ], \label{DNdelay3}\\
\dot{c}_{N-1,k}^{N-1}(t) = i\sqrt{N-1}\gamma_{1}  e^{-i\delta_{1}t} c_{N-2,k}^{N-2}(t) \notag \\
~~~~~~~~~~~~~~~ -\kappa  \left [c_{N-1,k}^{N-1}(t) - e^{i\Delta_0\tau} c_{N-1,k}^{N-1}(t-\tau) \right ], \label{DNdelay4}
\end{numcases}
\end{subequations}
\end{small}%
where $j =2,3,\cdots, N-2$. We define two matrices
\begin{scriptsize}
\begin{equation}
\begin{aligned}
  &A(t,\gamma_1,\cdots,\gamma_{N-1},\delta_1,\cdots,\delta_{N-1})\triangleq\\
     & \begin{bmatrix}
    0 & i\gamma_{N-1}  e^{i\delta_{N-1}t}  &\cdots & 0 &0\\
    i\gamma_{N-1} e^{-i\delta_{N-1}t} & -\kappa  &\cdots & 0 &0\\
    \vdots & \vdots & \ddots  &  \vdots & \vdots\\
    0 & 0  & \cdots   & i\sqrt{N-1}\gamma_{1}  e^{-i\delta_{1}t} & -\kappa \\
  \end{bmatrix},\\
  \end{aligned}
\end{equation}
\end{scriptsize}%
and
\begin{small}
\begin{equation}
\begin{aligned}
  B &\triangleq \begin{bmatrix}
    0 & 0 & 0 &\cdots  &0\\
    0 & \kappa e^{i\Delta_0\tau} & 0 &\cdots  &0 \\
    \vdots & \vdots &\vdots & \ddots  & \vdots  \\
    0 &  0 & 0 & \cdots  &\kappa e^{i\Delta_0\tau} \\
  \end{bmatrix}.\\
  \end{aligned}
\end{equation}
\end{small}%
Clearly, $A(t)$ is time-varying, and the constant matrix $B$ is determined by the round trip delay $\tau$ and the coupling strength $\kappa$ between the waveguide and cavity.
With the aid of these two matrices, Eq.~(\ref{con:DetuneNlevelDelay}) can be rewritten in a more compact form as
\begin{small}
\begin{subequations} \label{con:LinearControlDelay}
\begin{numcases}{}
\dot{X}(t) = A(t) X(t) + B X(t-\tau),\\
X(t) = \varphi(t), ~~\forall t\in [-\tau,0],
\end{numcases}
\end{subequations}
\end{small}%
where $\varphi(t) \equiv \left [1, 0,\cdots, 0 \right ]^{\top} $  for all $t\in [-\tau,0]$.

\begin{mypro} \label{Trapstate}
When $\gamma_{1}, \gamma_{2},\cdots,\gamma_{N-1} \neq 0$, there is no such $X(t)$ that $\lim_{t\rightarrow \infty}\dot{X}(t) = 0$ and $\lim_{t\rightarrow \infty}X(t) \neq 0$.
\end{mypro}

\begin{proof}
Notice that
\begin{small}
\begin{equation}
\begin{aligned}
  &A(t)+B =\\
   & \begin{bmatrix}
    0 & i\gamma_{N-1}  e^{i\delta_{N-1}t}  &\cdots  & 0 &0\\
    \vdots & \vdots &\ddots  &  \vdots & \vdots\\
    0 & 0  & \cdots  & i\sqrt{N-1}\gamma_{1}  e^{-i\delta_{1}t} & \kappa(e^{i\Delta_0\tau}-1) \\
  \end{bmatrix}.\\
  \end{aligned}
\end{equation}
\end{small}%
Denote the determinant $|A(t)+B| = D_N$. Clearly, $D_N \neq 0$ for all $t\geq0$ when $\gamma_{1}, \gamma_{2},\cdots,\gamma_{N-1} \neq 0$. 
Assume that   $\lim_{t\rightarrow \infty}X(t) = \bar{X}\neq 0 $. Then by Eq.~ \eqref{con:LinearControlDelay} we get $\lim_{t\rightarrow \infty}\dot{X}(t) = \lim_{t\rightarrow \infty}(A(t)+B)\bar{X}$. However, as $D_N \neq 0$ for all $t\geq0$, $\lim_{t\rightarrow \infty}\dot{X}(t)=0$ if and only if $\bar{X}=0$, which contradicts the condition that $\bar{X}\neq 0$.
\end{proof}

Notice that when an amplitude $c_{\psi}(t)$ of a quantum state $|\psi\rangle$ satisfies $\lim_{t\rightarrow\infty} \dot{c}_{\psi}(t) = 0$, then the state will not oscillate. We have the following proposition.

\begin{mypro} \label{OscillationProperty}
The oscillation property of the quantum feedback control system described in Eq.~(\ref{con:controleq2Detune}) is equivalent to the oscillation property of the amplitude $c_{1,k}^1(t)$ under the condition that $\gamma_{1}, \gamma_{2},\cdots,\gamma_{N-1} \neq 0$.
\end{mypro}

\begin{proof}
(1)~As proved in Proposition~\ref{Trapstate}, the condition $\lim_{t\rightarrow \infty}\dot{X}(t) = 0$ induces that $\lim_{t\rightarrow \infty}X(t) = 0$, and obviously $\lim_{t\rightarrow \infty}c_{1,k}^1(t) = 0$.

(2)~When $c_{1,k}^1(t)$ oscillates,  $\lim_{t\rightarrow \infty}c_{1,k}^1(t)$ does not exist. Then $c_0(t)$ oscillates according to Eq.~(\ref{Ddelaymodel1}).
\end{proof}

\subsection{Example: feedback control for a three-level atom with detunings}
\
\newline
For a $\Xi$-type three-level atom with detunings, Eq.~\eqref{con:controleq2Detune} becomes
\begin{small}
\begin{subequations} \label{con:DetunethreelevelDelay}
\begin{align}
&\dot{c}_0(t) =  i\gamma_{2} e^{i\delta_{2}t} c_{1,k}^{1}(t),   \label{Dthreedelay1}\\
&\dot{c}_{1,k}^1(t) =  i\gamma_2 e^{-i\delta_{2}t}c_0(t) + i \gamma_1 e^{i\delta_{1}t} c_{2,k}^2(t) \notag\\
&~~~~~~~~~~~-\kappa \left [c_{1,k}^1(t) -e^{i\Delta_0\tau} c_{1,k}^1(t-\tau) \right ], \label{Dthreedelay2b}\\
&\dot{c}_{2,k}^2(t) = i\gamma_1 e^{-i\delta_{1}t}c_{1,k}^1(t)
-  \kappa \left [ c_{2,k}^2(t) -e^{i\Delta_0\tau} c_{2,k}^2(t-\tau)\right ], \label{Dthreedelay4}\\
&\dot{c}_{1,k}^0(t,k) =i G^*(k,t) c_{1,k}^1(t)  + i\gamma_1  e^{i\delta_{1}t}c_{2,k}^1(t,k), \label{Dthreedelay3}\\
&\dot{c}_{2,k}^1(t,k)  = i\gamma_1 e^{-i\delta_{1}t}c_{1,k}^0(t,k) + i G^*(k,t) c_{2,k}^2(t)\notag\\
&~~~~~~~~~~~~~~-2\kappa \left [c_{2,k}^1(t,k) - e^{i\Delta_0\tau}c_{2,k}^1(t-\tau,k)\right ],\label{Dthreedelay5}\\
&\dot{c}_{2,k}^0(t,k_1,k_2) = iG^*(k_2,t) c_{2,k}^1(t,k_1) +  iG^*(k_1,t) c_{2,k}^1(t,k_2). \label{Dthreedelay6}
\end{align}
\end{subequations}
\end{small}%
Applying the Laplace transform to Eqs.~(\ref{Dthreedelay1}-\ref{Dthreedelay4}) we get
\begin{small}
\begin{subequations} \label{con:Laplace1to3b}
\begin{align}
&sC_0(s) -1  =  i\gamma_{2}  C_{1,k}^{1}(s-i\delta_2),   \label{Lapthreedelay1}\\
&sC_{1,k}^1(s) =  i\gamma_2 C_0(s+i\delta_2) + i \gamma_1 C_{2,k}^2(s-i\delta_1) \notag\\
&~~~~~~~~~~~~~~~-\kappa \left [C_{1,k}^1(s) -e^{i\Delta_0\tau} e^{-s\tau}C_{1,k}^1(s)\right ], \label{Lapthreedelay2}\\
&sC_{2,k}^2(s) = i\gamma_1 C_{1,k}^1(s+i\delta_1)  - \kappa[ C_{2,k}^2(s) -e^{i\Delta_0\tau} e^{-s\tau} C_{2,k}^2(s)]. \label{Lapthreedelay400}
\end{align}
\end{subequations}
\end{small}%

Next, we consider a simplified case
when $\tau \ll 1$ and $e^{-s\tau}\approx 1$. 
We denote the amplitudes in Eq.~(\ref{con:DetunethreelevelDelay}) in this parameter setting, i.e., $c_0(t)$, $c_{j,k}^m(t)$ as  $\breve{c}_0(t)$, $\breve{c}_{j,k}^m(t)$, and the Laplace transform in Eq.~(\ref{con:Laplace1to3}), i.e., $C_0(s)$, $C_{j,k}^m(s)$ as $\breve{C}_0(s)$, $\breve{C}_{j,k}^m(s)$, respectively.
Then the dynamics in Eq.~(\ref{con:DetunethreelevelDelay}) and Eq.~\eqref{con:Laplace1to3b} can be approximately solved as
\begin{small}
\begin{equation} \label{con:detuneC1k1s}
\begin{aligned}
&s\breve{C}_{1,k}^1(s) =   \frac{ -\gamma_{2}^2  \breve{C}_{1,k}^{1}(s) + i\gamma_2}{s+i\delta_2}\\
&- \frac{\gamma_1^2}{s-i\delta_1+ \kappa (1- e^{i\Delta_0\tau})} \breve{C}_{1,k}^1(s) -\kappa (1- e^{i\Delta_0\tau})\breve{C}_{1,k}^1(s) ,
\end{aligned}
\end{equation}
\end{small}%
and
\begin{small}
\begin{equation} \label{con:detuneC1k1s2}
\begin{aligned}
\breve{C}_{1,k}^1(s) =   \frac{\frac{i\gamma_2}{s+i\delta_2}}{\left [s+\kappa (1- e^{i\Delta_0\tau}) + \frac{\gamma_{2}^2}{s+i\delta_2} + \frac{\gamma_1^2}{s-i\delta_1+ \kappa (1- e^{i\Delta_0\tau})}\right]}.
\end{aligned}
\end{equation}
\end{small}%

In particular, when $\Delta_0\tau = 2n\pi$, $1-e^{i\Delta_0\tau} = 0$. In this case Eq.~ \eqref{con:detuneC1k1s2} reduces to
\begin{small}
\begin{equation} \label{con:detuneC1k1s22npi}
\begin{aligned}
\left[s + \frac{\gamma_{2}^2}{s+i\delta_2} + \frac{\gamma_1^2}{s-i\delta_1}\right ]\breve{C}_{1,k}^1(s) =   \frac{i\gamma_2}{s+i\delta_2} .
\end{aligned}
\end{equation}
\end{small}%
Consequently, Eq.~\eqref{con:detuneC1k1s} becomes
\begin{small}
\begin{equation} \label{con:detuneC1k1s22npi00}
\begin{aligned}
&~~~~\breve{C}_{1,k}^1(s) \\
&=  \frac{i\gamma_2(s-i\delta_1)}{s^3 +i(\delta_2-\delta_1)s^2 +(\gamma_1^2 + \gamma_2^2 +\delta_1\delta_2)s + i(\delta_2\gamma_1^2 - \delta_1\gamma_2^2)}.
\end{aligned}
\end{equation}
\end{small}%
When $\delta_1 = \delta_2 = 0$, Eq.~\eqref{con:detuneC1k1s22npi00} reduces to Eq.~(\ref{con:C1k1sZero}). We have the following result.

\begin{mypro} \label{Detuneproperty}
When $\Delta_0\tau = 2n\pi$ with $\tau \ll 1$ and $\frac{\delta_2}{\delta_1} = \frac{\gamma_2^2}{\gamma_1^2}$, $\breve{c}_{1,k}^1(t)$ oscillates persistently.
\end{mypro}

\begin{proof}
When $\Delta_0\tau = 2n\pi$ and $\frac{\delta_2}{\delta_1} = \frac{\gamma_2^2}{\gamma_1^2}$, in Eq.~(\ref{con:detuneC1k1s22npi00}) $\delta_2\gamma_1^2 - \delta_1\gamma_2^2 = 0$. Hence,
\begin{small}
\begin{equation} \label{con:detuneC1k1s22npiproof}
\begin{aligned}
&\breve{C}_{1,k}^1(s) =  \frac{i\gamma_2(s-i\delta_1)}{s^3 +i(\delta_2-\delta_1)s^2 +(\gamma_1^2 + \gamma_2^2 +\delta_1\delta_2)s }\\
& = \frac{A_0}{s} + \frac{A_1s+A_2}{s^2 +i(\delta_2-\delta_1)s +(\gamma_1^2 + \gamma_2^2 +\delta_1\delta_2) }\\
&= \frac{A_0}{s} + \frac{A_1s+A_2}{\left [s+\frac{i(\delta_2-\delta_1)}{2}\right]^2 + \left[\gamma_1^2 + \gamma_2^2 + \left(\frac{\delta_1+\delta_2}{2}\right)^2\right ]},
\end{aligned}
\end{equation}
\end{small}%
where $A_0 =\frac{\gamma_2\delta_1}{\gamma_1^2 + \gamma_2^2 + \delta_1\delta_2}$, $A_1$ and $A_2$ are nonzero constant numbers whose specific values are irrelevant. Due to  the second item of the last line of Eq.~(\ref{con:detuneC1k1s22npiproof}),  $\breve{c}_{1,k}^1(t)$ persistently oscillates around $\frac{\gamma_2\delta_1}{\gamma_1^2 + \gamma_2^2 + \delta_1\delta_2}$.
\end{proof}

\begin{remark}
The conclusion in Proposition~\ref{Detuneproperty} means that when there are detunings between the multi-level atom and cavity, there can be an atomic oscillating state. However, this is unlikely to occur because the condition $\frac{\delta_2}{\delta_1} = \frac{\gamma_2^2}{\gamma_1^2}$ is difficult to be satisfied in practice.
\end{remark}

\subsection{Stability of the $N$-level system}

Generalizing the definition of the exponential stability in  the real domain \cite{kharitonov2005exponential}, we give its definition in the complex space.
\newtheorem{myDef}{Definition}
\begin{myDef} \label{defstable}
 The system (\ref{con:LinearControlDelay}) is exponentially stable if there exist $\chi \geq 1$ and $\alpha > 0$ such that for every solution $X(t)$, the following exponential estimate holds:
\begin{small}
\begin{equation} \label{con:expon}
\begin{aligned}
\|X(t)\| \leq \chi  e^{-\alpha t} |\varphi|_{\tau}^*,
\end{aligned}
\end{equation}
\end{small}%
where $\|\cdot\|$ represents an arbitrary vector norm and $|\varphi|_{\tau}^* = \max_{t\in[-\tau,0]}\{\|\varphi(t)\|\}$.
\end{myDef}
\begin{remark}
In our system, $\varphi(t) \equiv 1$ for $t\in [-\tau, 0]$ according to the initial condition in Eq.~(\ref{con:LinearControlDelay}).
\end{remark}

\begin{remark}
The distributions of the photons between the cavity and waveguide of the coherent feedback network in Fig.~\ref{fig:Nlevel} can be studied by means of the exponential stability. Specifically, if the system is exponentially stable, there will be no photons in the cavity eventually, thus finally there is a multi-photon state in the waveguide. If the system oscillates rather than being exponentially stable, the photons are always being exchanged between the waveguide and cavity.
\end{remark}

Now we rewrite Eq.~(\ref{con:statevector}) by separating the real and imaginary parts as:
\begin{scriptsize}
\begin{equation} \label{con:Seperatestatevector}
\begin{aligned}
\tilde{X}(t) = [\bar{c}_0(t), \check{c}_0(t), \bar{c}_{1,k}^1(t),\check{c}_{1,k}^1(t),\cdots, \bar{c}_{N-1,k}^{N-1}(t),\check{c}_{N-1,k}^{N-1}(t)]^{\top}\in \mathbf{R}^{2N},
\end{aligned}
\end{equation}
\end{scriptsize}%
where $\bar{c}_0(t)$ represents the real part of $c_0(t)$, and $\check{c}_0(t)$ represents its imaginary part,
similarly for that of $c_{j,k}^j(t)$  for all $j = 1,2,\cdots,N-1$.

Define matrices
\begin{small}
\begin{equation}
\begin{aligned}
  &\tilde{A}(t,\gamma_1,\cdots,\gamma_{N-1},\delta_1,\cdots,\delta_{N-1})\triangleq\\
   & \begin{bmatrix}
    \mathbf{0}  &\gamma_{N-1} \mathbf{R}_{N-1}  &\cdots & \mathbf{0}  &\mathbf{0} \\
    \gamma_{N-1} \mathbf{R}_{N-1}  &-\kappa  I  &\cdots & \mathbf{0}  &\mathbf{0} \\
    \vdots  &\vdots & \ddots  & \vdots &\vdots\\
    \mathbf{0}  & \mathbf{0}   & \cdots   & \sqrt{N-1}\gamma_{1} \mathbf{R}_1 & -\kappa  I  \\
  \end{bmatrix},\\
  \end{aligned}
\end{equation}
\end{small}%
where
\begin{small}
\begin{equation} \label{Eq. R_j}
\begin{aligned}
  &\mathbf{R}_j\triangleq \begin{bmatrix}
   \sin(\delta_jt) & -\cos(\delta_jt)\\
   \cos(\delta_jt) & \sin(\delta_jt)
  \end{bmatrix},\\
  \end{aligned}
\end{equation}
\end{small}%
and
\begin{small}
\begin{equation} \label{con:Btilde}
\begin{aligned}
  \tilde{ B} &\triangleq \begin{bmatrix}
    0 & 0 & 0 &\cdots  &0\\
    0 & \kappa \mathbf{P}(\tau) & 0 &\cdots  &0 \\
    \vdots & \vdots &\vdots & \ddots  & \vdots  \\
    0 &  0 & 0 & \cdots  &\kappa \mathbf{P}(\tau) \\
  \end{bmatrix},\\
  \end{aligned}
\end{equation}
\end{small}%
where
\begin{small}
\begin{equation}
\begin{aligned}
  &\mathbf{P}(\tau)\triangleq\begin{bmatrix}
   \cos(\Delta_0\tau) & -\sin(\Delta_0\tau)\\
   \sin(\Delta_0\tau) & \cos(\Delta_0\tau)
  \end{bmatrix}.\\
  \end{aligned}
\end{equation}
\end{small}%
Then the system in Eq.~(\ref{con:LinearControlDelay}) can be written as
\begin{small}
\begin{subequations} \label{con:RealLinearControlDelay}
\begin{numcases}{}
\dot{\tilde{X}}(t) = \tilde{A}(t) \tilde{X}(t) + \tilde{B} \tilde{X}(t-\tau),\\
\tilde{X}(t) = \tilde{\varphi}(t), ~~\forall t\in [-\tau,0],
\end{numcases}
\end{subequations}
\end{small}%
in the  real state-space configuration.

\begin{remark}
When $\delta_j = 0$, the energy difference between arbitrary neighboring levels is the same, which is $\omega_c$. In this case $\tilde{A}$ in Eq.~\eqref{con:RealLinearControlDelay} is time-invariant because $\mathbf{R}_j$ in Eq.~\eqref{Eq. R_j} now reduces to $\mathbf{R}_j= \begin{bmatrix} 0&-1\\1& 0\end{bmatrix}$.
\end{remark}

In the following, we investigate the dynamics of the quantum coherent feedback  according to whether or not the detunings $\delta_j=0$.

a)~ The case of $\delta_j = 0$.

When $\delta_j = 0$, applying the Laplace transform to Eq.~(\ref{con:RealLinearControlDelay}) gives
\[
s\tilde{X}(s) - \tilde{X}(0) = \tilde{A}\tilde{X}(s) + \tilde{B}\tilde{X}(s)e^{-s\tau}.
\]
Hence
\begin{small}
\begin{equation} \label{con:Xs}
\begin{aligned}
\tilde{X}(s) =  (sI -\tilde{B}e^{-s\tau} - \tilde{A})^{-1} \tilde{X}(0),
\end{aligned}
\end{equation}
\end{small}%
where $I$ is the $2N\times2N$  identity matrix.
It can be seen that
\begin{small}
\begin{equation} \label{con:YJs}
\begin{aligned}
&~~~~(sI -\tilde{B}e^{-s\tau} - \tilde{A})^{-1} \\
&=  \begin{bmatrix}
   s I   & -\gamma_{N-1} \mathbf{R}  &\cdots   &\mathbf{0}  &\mathbf{0} \\
   -\gamma_{N-1} \mathbf{R} & \mathbf{D}   &\cdots &\mathbf{0}  &\mathbf{0} \\
   \vdots & \vdots &\ddots &\vdots  &\vdots \\
   \mathbf{0}  & \mathbf{0}  &\cdots   & -\sqrt{N-1}\gamma_{1} \mathbf{R} & \mathbf{D}
  \end{bmatrix}^{-1}\\
&=\frac{1}{|sI -\tilde{B}e^{-s\tau} - \tilde{A}|} (sI -\tilde{B}e^{-s\tau} - \tilde{A})^*,
\end{aligned}
\end{equation}
\end{small}%
where the superscript $*$ represents the adjugate operation, $\mathbf{R}=\begin{bmatrix} 0&-1\\1& 0\end{bmatrix}$, and $\mathbf{D}  =  (s+\kappa) I   - \kappa e^{-s\tau} \mathbf{P}(\tau)$.

\begin{mypro} \label{Nleveloscillate}
When $\Delta_0\tau = 2n\pi$ and $\delta_j = 0$, $\tilde{X}(t)$ oscillates with the frequency determined by the coupling strengths between the multi-level atom and the cavity.
\end{mypro}

\begin{proof}
The dynamics of the quantum system in Eq.~(\ref{con:RealLinearControlDelay}) is determined by its poles, which are roots of the equation $|sI -\tilde{B}e^{-s\tau} - \tilde{A}|=0$. When $\Delta_0\tau = 2n\pi$, $\mathbf{P}(\tau)= \begin{bmatrix}
   1 & 0\\
   0 & 1
  \end{bmatrix}$ in Eq.~(\ref{con:Btilde}), $\mathbf{R}_j= \begin{bmatrix} 0&-1\\1& 0\end{bmatrix}$, and $\tau \approx  0$ because $\Delta_0 \gg 1$ in practical systems. Then in Eq.~(\ref{con:YJs}), $\mathbf{D}  \approx  s I   + \kappa( I   - \mathbf{P}(\tau)) = s I$.
After some tedious calculation we find  the determinant
\begin{small}
\begin{equation} \label{con:detpole}
\begin{aligned}
&~~~~|sI -\tilde{B}e^{-s\tau} - \tilde{A}| \\
&=|s\mathbf{D}  + \gamma_{N-1}^2  I  | |\mathbf{D} ^2 +2\gamma_{N-2}^2  I  | \cdot \cdots \cdot |\mathbf{D} ^2 +(N-1)\gamma_{1}^2  I  |\\
&=(s^2+\gamma_{N-1}^2)^2 (s^2+2\gamma_{N-2}^2)^2 \cdot \cdots \cdot (s^2+(N-1)\gamma_{1}^2)^2.
\end{aligned}
\end{equation}
\end{small}%
Obviously, all the roots of the equation $|sI -\tilde{B}e^{-s\tau} - \tilde{A}|=0$ are on the imaginary axis of the complex plane, which means that $\tilde{X}(t)$ oscillates with the frequency determined by $\gamma_1,\gamma_2,\cdots,\gamma_{N-1}$.
\end{proof}

\begin{mypro} \label{TACLya}
\cite{mondie2005exponential,kuang2018stability} For the real time-delayed system (\ref{con:RealLinearControlDelay}) with $\delta_j = 0$, if there exist real  positive definite matrices $P$, $Q$ and a positive constant $\beta$ such that the inequality
\begin{small}
\begin{equation} \label{con:Lya}
\begin{aligned}
 \mathcal{M}(P,Q) + 2\beta \mathcal{N}(P) <0
  \end{aligned}
\end{equation}
\end{small}
holds, where
\begin{small}
\begin{equation} \label{MNdef}
\begin{aligned}
  \mathcal{M}(P,Q)&= \begin{bmatrix}
   P \tilde{A} +\tilde{A}^\top  P+Q  & P\tilde{B}\\
   \tilde{B}^\top  P & -e^{-2\beta\tau}Q
  \end{bmatrix},\\
  \mathcal{N}(P)&=\begin{bmatrix} I\\0 \end{bmatrix}P \begin{bmatrix} I &0 \end{bmatrix} =\begin{bmatrix}
   P  & 0\\
   0 & 0
  \end{bmatrix},
  \end{aligned}
\end{equation}
\end{small}%
then
\begin{small}
\begin{equation} \label{LyaNeq}
\begin{aligned}
\|\tilde{X}(t,\tilde{\varphi}(t))\| \leq \sqrt{\frac{\alpha_2}{\alpha_1}} e^{-\beta t} \|\tilde{\varphi}(t))\|_\tau,
  \end{aligned}
\end{equation}
\end{small}%
where the positive constants $\alpha_1$ and $\alpha_2$ are defined as
\begin{subequations} \label{con:alpha12}
\begin{numcases}{}
\alpha_1 = \lambda_{min}(P),\\
\alpha_2 = \lambda_{max}(P) + \tau \lambda_{max}(Q).
\end{numcases}
\end{subequations}
\end{mypro}

\newtheorem{corol}{Corollary}

\begin{remark}
According to Proposition~\ref{Nleveloscillate}, when $\Delta_0\tau = 2n\pi$ and $\delta_j = 0$, the system oscillates. In this case, the relationship in Eq.~(\ref{con:Lya}) cannot be satisfied.
\end{remark}

A generalized version of Proposition \ref{TACLya} will be given later.

\medskip

b)~ The case of $\delta_j \neq 0$.

When $\delta_j \neq 0$, denote
\begin{small}
\begin{equation}
\begin{aligned} \label{con:timevarydelta}
&\tilde{A}(t,\gamma_1,\cdots,\gamma_{N-1},\delta_1,\cdots,\delta_{N-1})\\
=& \tilde{A}_0 + \Upsilon(t,\gamma_1,\cdots,\gamma_{N-1},\delta_1,\cdots,\delta_{N-1}),
\end{aligned}
\end{equation}
\end{small}%
where
\begin{small}
\begin{equation}
\begin{aligned}
  &\tilde{A}_0 =
   \begin{bmatrix}
    \mathbf{0}  &\gamma_{N-1} \mathbf{R}  &\cdots  & \mathbf{0}  &\mathbf{0} \\
    \gamma_{N-1} \mathbf{R} &-\kappa  I    &\cdots & \mathbf{0}  &\mathbf{0} \\
    \vdots  &\vdots & \ddots &\vdots &\vdots\\
    \mathbf{0}  & \mathbf{0}   & \cdots   & \sqrt{N-1}\gamma_{1} \mathbf{R} & -\kappa  I  \\
  \end{bmatrix},\\
  \end{aligned}
\end{equation}
\end{small}%
and
\begin{small}
\begin{equation}
\begin{aligned}
  &\Upsilon(t,\gamma_1,\cdots,\gamma_{N-1},\delta_1,\cdots,\delta_{N-1})=\\
   & \begin{bmatrix}
    \mathbf{0}  &\gamma_{N-1} \tilde{\mathbf{R}}_{N-1}  &\cdots  & \mathbf{0}  &\mathbf{0} \\
    \vdots  &\vdots & \ddots &\vdots &\vdots\\
    \mathbf{0}  & \mathbf{0}   & \cdots & \sqrt{N-1}\gamma_{1} \tilde{\mathbf{R}}_1 &\mathbf{0} \\
  \end{bmatrix},\\
  \end{aligned}
\end{equation}
\end{small}%
with $\mathbf{R}=\begin{bmatrix} 0&-1\\1& 0\end{bmatrix}$, $\mathbf{R}_j= \begin{bmatrix}
   \sin(\delta_jt) & -\cos(\delta_jt)\\
   \cos(\delta_jt) & \sin(\delta_jt)
  \end{bmatrix}$, and $\tilde{\mathbf{R}}_{j} = \mathbf{R}_{j}-\mathbf{R}$.

The induced $2$-norm of $\Upsilon(t,\gamma_1,\cdots,\gamma_{N-1},\delta_1,\cdots,\delta_{N-1})$ is the square root of the maximum eigenvalue of $\Upsilon^\top \Upsilon$. It can be found that
\begin{scriptsize}
\begin{equation}
\begin{aligned}
&\Upsilon^\top (t,\gamma_1,\cdots,\gamma_{N-1},\delta_1,\cdots,\delta_{N-1}) \Upsilon(t,\gamma_1,\cdots,\gamma_{N-1},\delta_1,\cdots,\delta_{N-1})\\
  &=\begin{bmatrix}
  \gamma_{N-1}^2 \tilde{\mathbf{R}}_{N-1}^\top \tilde{\mathbf{R}}_{N-1}    &\cdots    &\mathbf{0}  &\cdots &\mathbf{0} \\
  \vdots  &\ddots &\vdots &\ddots &\vdots\\
    \mathbf{0}    &\cdots &\mathfrak{R}_j
  &\cdots  &\mathbf{0}  \\
  \vdots  &\ddots &\vdots  &\ddots &\vdots\\
  \mathbf{0}   &\cdots &\mathbf{0}  &\cdots  &(N-1)\gamma_1^2 \tilde{\mathbf{R}}_{1}^\top \tilde{\mathbf{R}}_{1}
  \end{bmatrix},
  \end{aligned}
\end{equation}
\end{scriptsize}%
with $\mathfrak{R}_j = j\gamma_{N-j}^2 \tilde{\mathbf{R}}_{N-j}^\top \tilde{\mathbf{R}}_{N-j} + (j+1)\gamma_{N-j-1}^2 \tilde{\mathbf{R}}_{N-j-1}^\top \tilde{\mathbf{R}}_{N-j-1}$ and
\begin{small}
\begin{equation}
\begin{aligned}
\tilde{\mathbf{R}}_{j}^\top \tilde{\mathbf{R}}_{j} = [\sin^2(\delta_jt) + (1-\cos(\delta_jt))^2]  I   = 2(1-\cos(\delta_jt)) I  .
  \end{aligned}
\end{equation}
\end{small}%
Hence, the induced $2$-norm of $\Upsilon(t,\gamma_1,\cdots,\gamma_{N-1},\delta_1,\cdots,\delta_{N-1})$ reads
\begin{small}
\begin{equation}
\begin{aligned} \label{con:norm}
&\|\Upsilon(t,\gamma_1,\cdots,\gamma_{N-1},\delta_1,\cdots,\delta_{N-1}) \|_2 \\
=&\max_j \sqrt{ \sum_{\mathfrak{j}=j}^{j+1} 2(N-\mathfrak{j}) \gamma_{\mathfrak{j}}^2  (1-\cos(\delta_{\mathfrak{j}}t))},
  \end{aligned}
\end{equation}
\end{small}%
which is finite when $1\leq j \leq N-1$, and is determined by the coupling strength and detunings between the atom and the cavity.

The following result generalizes \textbf{Proposition}~{\ref{TACLya}} by allowing nonzero detunings.

\begin{mypro} \label{TACLyaUncertainty}
For the quantum control system described by Eq.~(\ref{con:RealLinearControlDelay}), if there exist real  positive-definite matrices $\tilde{P}$, $\tilde{Q}$, and a positive constant $\beta$ such that
\begin{small}
\begin{equation} \label{con:LyatimeVary}
\begin{aligned}
 &\mathcal{M}(\tilde{P},\tilde{Q}) + 2 \lambda_{max}(\tilde{P}) \|\Upsilon(t,\gamma_1,\cdots,\gamma_{N-1},\delta_1,\cdots,\delta_{N-1}) \|_2 I\\
 &+ 2\beta \mathcal{N}(\tilde{P}) <0,
  \end{aligned}
\end{equation}
\end{small}%
where $\mathcal{M}$ and $\mathcal{N}$ are those defined in Eq.~(\ref{MNdef}), then $\tilde{X}(t)$ is exponentially stable.
\end{mypro}
\begin{proof}
We prove this by constructing a Lyapunov-Krasovkii function, and the details are given in \textbf{Appendix~\ref{Sec:ProofLyaApend}}.
\end{proof}
\begin{remark}
When $\delta_j = 0$ for $j =1,2,\cdots, N-1$, $\|\Upsilon(t,\gamma_1,\cdots,\gamma_{N-1},\delta_1,\cdots,\delta_{N-1}) \|_2 = 0$ in Eq.~(\ref{con:norm}). That is, in the resonant case \textbf{Proposition}~\ref{TACLyaUncertainty} reduces to \textbf{Proposition}~\ref{TACLya}. Non-zero detunings introduce the second item on the RHS of Eq.~(\ref{con:LyatimeVary}), which is absent in Eq.~(\ref{con:Lya}) for the resonant case. Thus, the existence of detunings between the multi-level atom and cavity makes it more difficult for the linear quantum control system to be exponentially stable.   
\end{remark}

We take the four-level atom as an example for demonstration, and compare the feedback control performance between the cases whether there are atom-cavity detunings or not. As illustrated in Fig.~\ref{fig:Fourlevel}, the simulations are taken based on the original evolution equation Eq.~(\ref{con:controleq2Detune}) without any approximation. We take the parameters as $\Delta_0 = 50$, $G_0 = 0.2$, $\gamma_j = 0.6$ for $j=1,2,3$. The solid lines represent the evolution of populations when $\delta_j = 0$ and the dashed lines represent the case $\delta_j = 1$. The comparison reveals that when there are detunings between the cavity and the multi-level atom, the coherent feedback control stability is worse, which agrees with  Proposition~\ref{TACLyaUncertainty} saying that the stability inequality is more difficult to be satisfied when there are detunings between the cavity and the multi-level atom.

\begin{figure}[h]
\centerline{\includegraphics[width=1\columnwidth]{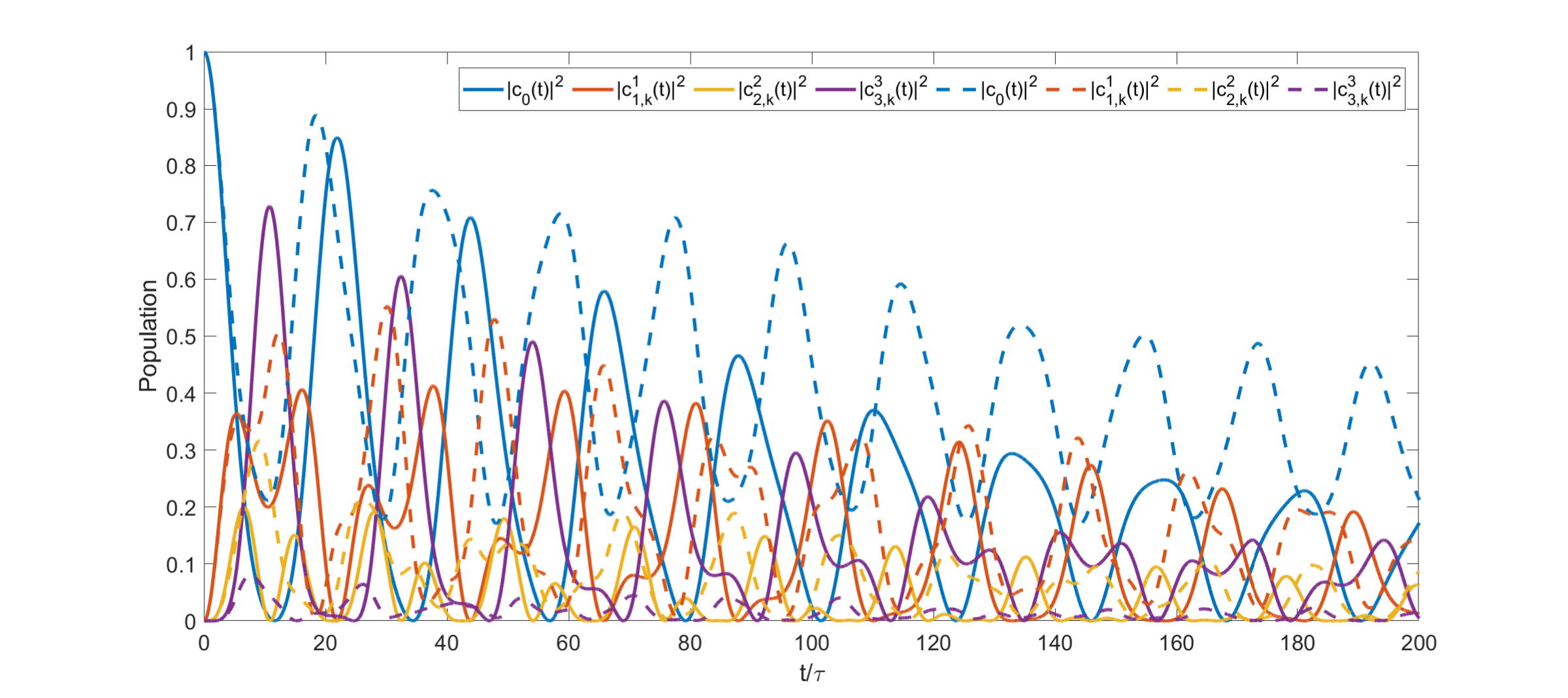}}
\caption{Feedback control performance when the four-level atom is resonant (the solid lines) or non-resonant (the dashed lines) with the cavity.} 
\label{fig:Fourlevel}
\end{figure}

\section{Coherent feedback control with multiple parallel waveguides} \label{Sec:Parallel}

Apart from the architecture that a single waveguide is coupled with a cavity to construct a feedback loop, in many practical optical systems, a cavity can be coupled to multiple parallel waveguides. Such configurations have abundant potential applications in photon routing~\cite{ParaWphotonRout}, creating correlated photonic states~\cite{WWaveInt,CorreMulWPRL,poulios2014quantum}, among others.
Therefore in this section, we  study the quantum coherent feedback network as shown in Fig.~\ref{fig:multiwaveguide} (a), which  generalizes the one in Fig.~\ref{fig:Nlevel} by allowing {\it multiple} waveguides. Specifically, the cavity is coupled to the inner red waveguide with coupling strengths $G^*(k,t)$ and $G(k,t)$, and  each waveguide is only coupled with its nearest neighbors. Thus the photons in the inner red waveguide can be transmitted into the outer black waveguides through the cascade coupling designated by coupling strengths $K_{12},K_{23},\cdots,K_{(W-1)W}$, and vice versa. This process can be equivalently represented in Fig.~\ref{fig:multiwaveguide} (b), where the parallel waveguides are coupled to each other designated by coupling strengths $K_{12},K_{23},\cdots,K_{(W-1)W},K_{W(W-1)},\cdots,K_{32},K_{21}$, thus forming a coherent feedback network so that photons are transmitted back and forth between the cavity and parallel waveguides.

\begin{figure}[h]
\centerline{\includegraphics[width=1\columnwidth]{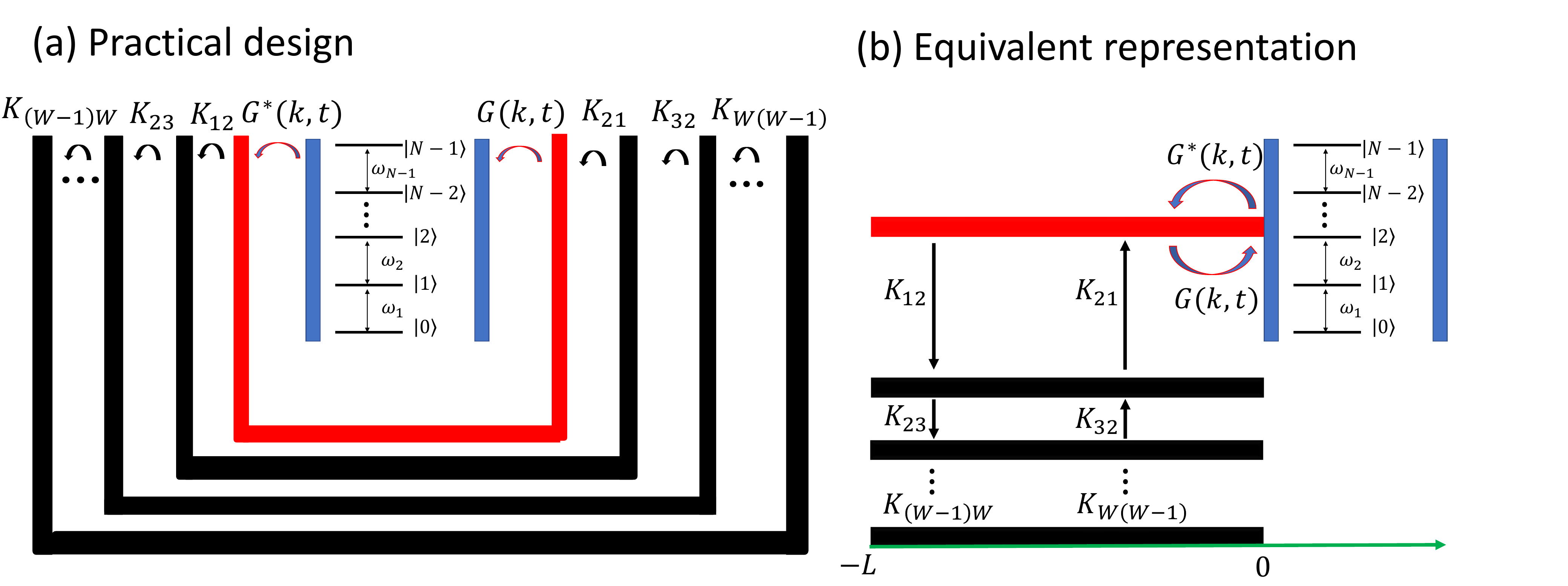}}
\caption{An $N$-level atom in a cavity coupled with multiple parallel waveguides.}
	\label{fig:multiwaveguide}
\end{figure}

According to the studies in Ref.~\cite{ding2022quantum} and Section~\ref{Sec:Model} above, when $\delta_n \ll \omega_c$ and $\gamma_n \ll \omega_c$, once there are observable photons in a waveguide, its photonic modes will be Lorentzian around the central mode $\omega_c$, which means that the couplings among the waveguide array are largely determined by the central mode $\omega_c$ of the Lorentzian spectrum; see also  Refs.~\cite{poulios2014quantum,meinecke2013coherent,Parallelmultiphoton}. Due to this, we make the following assumption.
\begin{assumption} \label{waveguideCoupling}
The coupling strengths $K_{\mathbf{w}\mathbf{(w+1)}}$ and $K_{(\mathbf{w+1})\mathbf{w}}$ between the two neighboring waveguides are constant, for all $\mathbf{w}=1,2,\ldots W-1$.
\end{assumption}

Generalized from Eq.~(\ref{con:Hamiltonian}), the Hamiltonian of the system in  Fig.~\ref{fig:multiwaveguide} reads
\begin{small}
\begin{equation} \label{con:ParallelHam}
\begin{aligned}
H' = H_A  + \omega_c a^{\dag}a +\sum_{\mathbf{w}=1}^W\int \omega_kd_k^{\dag \mathbf{w}}d_k^{\mathbf{w}} \mathrm{d}k+ H_{I}, 
\end{aligned}
\end{equation}
\end{small}%
where the waveguide array Hamiltonian is a summation of all the waveguides with continuous modes. The interaction picture Hamiltonian reads~\cite{WWaveInt}
\begin{scriptsize}
\begin{equation} \label{con:HamIParawaveguide}
\begin{aligned}
&\tilde{H}_I' = -\sum_{n=1}^{N-1} \gamma_n \left(e^{-i\delta_nt}\sigma_n^-a^{\dag} + e^{i\delta_nt}\sigma_n^+a \right)\\
& - \int \left[G^*(k,t)ad^{\dag 1}_k + G(k,t) a^{\dag}d_k^1\right]
 -\sum_{\mathbf{w}=1}^{W-1} K_{\mathbf{w}(\mathbf{w}+1)} d_k^{\mathbf{w}} d_k^{\dag (\mathbf{w}+1)} \mathrm{d}k\\
& -\sum_{\mathbf{w}=2}^W K_{\mathbf{w}(\mathbf{w}-1)} d_k^{\mathbf{w}} d_k^{\dag (\mathbf{w}-1)} -\sum_{\mathbf{w}=1}^W \beta_{\mathbf{w}}  d_k^{\dag \mathbf{w}}d_k^{\mathbf{w}},\\
\end{aligned}
\end{equation}
\end{scriptsize}%
where the real number $\beta_{\mathbf{w}}$  is the propagation constant of the modes in the $\mathbf{w}$-th waveguide and  $\beta_{\mathbf{w}} d_k^{\dag \mathbf{w}}d_k^{\mathbf{w}} $ represents that the photon number in the $\mathbf{w}$-th waveguide is not changed. For simplicity, it is assumed that  $K_{\mathbf{w}(\mathbf{w}+1)} = K_{(\mathbf{w}+1)\mathbf{w}}^*$ for $\mathbf{w} = 1,2,\cdots,W-1$, as has been done in \cite{Parallelmultiphoton,ParalleWcouple,chiralmedia,ParallelCoupledModes}.

\subsection{Photon states in the parallel waveguides}
If there are $n_\mathbf{w}$ photons in the $\mathbf{w}$-th waveguide, then the $n_\mathbf{w}$-photon state in the $\mathbf{w}$-th waveguide can be represented as the linear superposition of the unnormalized basis states
\begin{equation} \label{con:wwaveguidephoton}
\begin{aligned}
|\{k_{\odot n_\mathbf{w}}^{\mathbf{w}}\}\rangle \triangleq
d^{\dag \mathbf{w}}_{k_{n_\mathbf{w}\odot n_\mathbf{w}}} \cdots d^{\dag \mathbf{w}}_{k_{1\odot n_\mathbf{w}}} |0\rangle,
\end{aligned}
\end{equation}
where $d^{\dag \mathbf{w}}_{k_p\odot n_{\mathbf{w}}}|0\rangle$ with $p=1,2,\cdots,n_{\mathbf{w}}$ creates a single photon of mode $k_p$ in the $\mathbf{w}$-th waveguide, and the symbol $\odot n_\mathbf{w}$ indicates there are altogether $n_{\bf w}$  photons in  the $\mathbf{w}$-th waveguide. Of course, $|\{k_{\odot n_\mathbf{w}}^{\mathbf{w}}\}\rangle$ is the vacuum state $\ket{0}$ if $n_{\mathbf{w}}=0$, namely no photons in the $\mathbf{w}$-th waveguide.  

Assume the total number of photons in the waveguide array is $N_W \triangleq \sum_{\mathbf{w} = 1}^{W }n_{\mathbf{w}}$.
Then the basis states for the photons in the waveguide array read
\begin{small}
\begin{equation} \label{con:MwaveguidePhoton}
\begin{aligned}
&~~~~|\{k_{\odot n_1}^{1}\}, \cdots, \{k_{\odot n_\mathbf{w}}^{\mathbf{w}}\}, \cdots,\{k_{\odot n_W}^{W}\}\rangle \\
&\triangleq  |\{k_{1\odot n_1}^{1}, \cdots,k_{n_1\odot n_1}^{1}\}, \cdots,\{k_{1\odot n_{\mathbf{w}}}^{\mathbf{w}}, \cdots,k_{n_\mathbf{w}\odot n_\mathbf{w}}^{\mathbf{w}}\} ,  \cdots,\\
&~~~~\{k_{1\odot n_W}^{W} \cdots,k_{n_W\odot n_W}^{W}\}\rangle ,\\
\end{aligned}
\end{equation}
\end{small}%
where $\{k_{1\odot n_{\mathbf{w}}}^{\mathbf{w}}, \cdots,k_{n_\mathbf{w}\odot n_\mathbf{w}}^{\mathbf{w}}\}$ represents that the basis states $ \{k_{\odot n_\mathbf{w}}^{\mathbf{w}}\}$ in the $\mathbf{w}$-th waveguide are composed with $n_\mathbf{w}$ photons.
Because of the exchange of photons between the neighboring waveguides, the photon state in the parallel waveguides is not simply the product of the states of different channels; instead it is a superposition of the  tensor-product basis states  given in Eq. \eqref{con:MwaveguidePhoton}.
\subsection{Quantum states of the feedback system}

For the quantum system in Fig.~\ref{fig:multiwaveguide}, when the atomic state is $|N-1-j\rangle$ and there are $m$ photons in the cavity, then there are $j-m$ photons in the parallel waveguides. Assume the number of photons in the $i$-th waveguide is $n_i$, $i=1,2,\cdots,W$, then the quantum system can be represented as~\cite{zhang2014analysis,Z17,ZD22,DZWW23}
\begin{scriptsize}
\begin{equation} \label{con:stateMwaveguide}
\begin{aligned}
&|\Psi(t)\rangle =\sum_{j=0}^{N-1} c_j^j(t)|N-1-j,j,\{0\}\rangle \\
&+ \sum_{j = 1}^{N-1} \sum_{m=0}^{j-1} \sum_{N_W = j-m}\int \cdots \int c_{j,\odot n_1\cdots\odot n_W}^{m}\left(t,k_{\odot n_1}^{1},\cdots,k_{\odot n_W}^{W}\right) \\
&~~~~\left|N-1-j,m,\{k_{\odot n_1}^{1}\}, \cdots,\{k_{\odot n_W}^{W}\}\right\rangle \mathrm{d}k_{1\odot n_1}^{1} \cdots \mathrm{d}k_{n_W\odot n_W}^{W},
\end{aligned}
\end{equation}
\end{scriptsize}%
where 
$c_{j,\odot n_1\cdots\odot n_W}^{m}$ is the amplitude of the state with overall $j$ excitations including $m$ photons in the cavity and $n_1,n_2,\cdots,n_W$ photons in the first, second, $\cdots$, $W$-th waveguide, respectively.
Obviously, $1 \leq \max\{n_{\mathbf{w}}\} \leq j-m$ and $N_{\mathbf{w}}\equiv \sum_{\mathbf{w} = 1}^M n_{\mathbf{w}} = j-m $. Initially, the atom is excited at $|N-1\rangle$ and $n_1 = n_2 = \cdots = n_W = 0$.

The normalization condition of the state in Eq.~\eqref{con:stateMwaveguide} is
\begin{small}
\begin{equation} \label{con:Norm}
\begin{aligned}
&\sum_{j=1}^{N-1} |c_j^j(t)|^2 +  \sum_{j = 1}^{N-1} \sum_{m=0}^{j-1} \sum_{N_W = j-m}\int \cdots \int\\
&|c_{j,\odot n_1\cdots\odot n_W}^{m}(t,k_{\odot n_1}^{1},\cdots,k_{\odot n_W}^{W})|^2 \mathrm{d}k_{1\odot n_1}^{1} \cdots \mathrm{d}k_{n_W\odot n_W}^{W}=1.
\end{aligned}
\end{equation}
\end{small}%

The population representing that there are $n$ photons in the $\mathbf{w}$-th waveguide can be calculated as
\begin{scriptsize}
\begin{equation} \label{con:Pwn}
\begin{aligned}
&P_n^{\mathbf{w}} = \sum_{j = n}^{N-1} \sum_{m=0}^{j-n} \sum_{n_1 + \cdots + n_{\mathbf{w}-1} +  n_{\mathbf{w}+1}  + n_W = j-m-n} \\
&\int \cdots \int \left|c_{j,\odot n_1\cdots\odot n_W}^{m}(t,k_{\odot n_1}^{1},\cdots,k_{\odot n_W}^{W})\right|^2 \mathrm{d}k_{1\odot n_1}^{1} \cdots \mathrm{d}k_{n_W\odot n_W}^{W},
\end{aligned}
\end{equation}
\end{scriptsize}%
because there can be $n$ photons in a single waveguide only when the atom is at the state $|N-1-j\rangle$ with $j = n,n+1,\cdots,N-1$.

Substituting Eq.~(\ref{con:stateMwaveguide}) into the Schr\"{o}dinger equation with the Hamiltonian in Eq.~(\ref{con:HamIParawaveguide}), we can get
\begin{tiny}
\begin{subequations} \label{con:ParallelControl}
\begin{align}
 &\dot{c}_j^j(t) = i\sqrt{j}\gamma_{N-j}  e^{-i\delta_{N-j}t} c_{j-1}^{j-1}(t)\notag\\
 &+ i\sqrt{j+1}\gamma_{N-j-1}  e^{i\delta_{N-j-1}t} c_{j+1}^{j+1}(t) -\kappa  [c_{j,k}^{j}(t) - e^{i\Delta_0\tau} c_{j,k}^{j}(t-\tau) ],  \label{photoncavity}\\
&\dot{c}_{j,\odot n_1\cdots\odot n_W}^{m}(t,k_{\odot n_1}^{1},\cdots,k_{\odot n_W}^{W}) \notag\\
&= i\sqrt{m}\gamma_{N-j}  e^{-i\delta_{N-j}t} c_{j-1,\odot n_1\cdots\odot n_W}^{m-1}(t,k_{\odot n_1}^{1},\cdots,k_{\odot n_W}^{W}) \notag\\
&  + i\sqrt{m+1}\gamma_{N-j-1}  e^{i\delta_{N-j-1}t} c_{j+1,\odot n_1\cdots\odot n_W}^{m+1}(t,k_{\odot n_1}^{1},\cdots,k_{\odot n_W}^{W}) \notag \\
&+i\sum_{p=1}^{n_1+1}\int G(k_p,t) c_{j,\odot n_1+1\cdots\odot n_W}^{m-1}(t,k_{1\odot n_1+1}^1,\cdots,k_{p\odot n_1+1}^1,\cdots,\notag\\
& k_{n_1+1\odot n_1+1}^1;k_{1\odot n_2}^2,\cdots,k_{n_2\odot n_2}^2;\cdots; k_{1\odot n_W}^W;\cdots,k_{n_W\odot n_W}^W) \mathrm{d}k_p \notag\\
&+i\sum_{p=1}^{n_1} G^*(k_p,t) c_{j,\odot n_1-1\cdots\odot n_W}^{m+1} \notag\\
&~~~~~~(t,k_{1\odot n_1-1}^1,\cdots,k_{p-1\odot n_1-1}^1,k_{p+1\odot n_1-1}^1,\cdots, k_{n_1-1\odot n_1-1}^1; \notag\\
&k_{1\odot n_2}^2,\cdots,k_{n_2\odot n_2}^2;\cdots; k_{1\odot n_W}^W;\cdots,k_{n_W\odot n_W}^W)\notag\\
&+ i\sum_{\mathbf{w} =1}^{W-1} K_{\mathbf{w} (\mathbf{w}+1)} c_{j,\odot n_1\cdots \odot n_{\mathbf{w}}+1\odot n_{\mathbf{w}+1}-1\cdots\odot n_W}^{m}\notag\\
&~~~~~~(t,k_{\odot n_1}^{1},\cdots,k_{\odot n_{\mathbf{w}}+1}^{\mathbf{w}},k_{\odot n_{\mathbf{w}+1}-1}^{\mathbf{w}+1},\cdots,k_{\odot n_W}^{W})\notag\\
&+ i\sum_{\mathbf{w} =2}^{W} K_{\mathbf{w}(\mathbf{w}-1)}
c_{j,\odot n_1\cdots \odot n_{\mathbf{w}-1}-1\odot n_{\mathbf{w}}+1\cdots\odot n_W}^{m}\notag\\
&~~~~~~(t,k_{\odot n_1}^{1},\cdots,k_{\odot n_{\mathbf{w}-1}-1}^{\mathbf{w}-1},k_{\odot n_{\mathbf{w}}+1}^{\mathbf{w}},\cdots,k_{\odot n_W}^{W})\notag\\
&+i\sum_{\mathbf{w} =1}^W n_{\mathbf{w}}\beta_{\mathbf{w}} c_{j,\odot n_1\cdots\odot n_W}^{m}(t,k_{\odot n_1}^{1},\cdots,k_{\odot n_W}^{W}), m>0,\label{CjkmParaW}\\
&\dot{c}_{j,\odot n_1\cdots\odot n_W}^0(t,k_{\odot n_1}^{1},\cdots,k_{\odot n_W}^{W}) \notag\\
&=  i\gamma_{N-j-1}   e^{i\delta_{N-j-1}t}  c_{j+1,\odot n_1\cdots\odot n_W}^{1}(t,k_{\odot n_1}^{1},\cdots,k_{\odot n_W}^{W}) \notag\\
&+i\sum_{p=1}^{n_1-1} G^*(k_p,t) c_{j,\odot n_1-1\cdots\odot n_W}^{1}\notag\\
&~~~~~~(t,k_{1\odot n_1}^1, \cdots, k_{p-1\odot n_1}^1,k_{p+1\odot n_1}^1,\cdots, k_{n_1\odot n_1}^1,k_{\odot n_2}^2 \cdot k_{\odot n_W}^W)  \notag\\
&+ i\sum_{\mathbf{w} =1}^{W-1} K_{\mathbf{w} (\mathbf{w}+1)} c_{j,k}^{0}(t,k_{\odot n_1}^{1},\cdots,k_{\odot n_{\mathbf{w}}+1}^{\mathbf{w}},k_{\odot n_{\mathbf{w}+1}-1}^{\mathbf{w}+1},\cdots,k_{\odot n_W}^{W})\notag\\
&+ i\sum_{\mathbf{w} =2}^{W} K_{\mathbf{w}(\mathbf{w}-1)}c_{j,k}^{0}(t,k_{\odot n_1}^{1},\cdots,k_{\odot n_{\mathbf{w}-1}-1}^{\mathbf{w}-1},k_{\odot n_{\mathbf{w}}+1}^{\mathbf{w}},\cdots,k_{\odot n_W}^{W})\notag\\
&+i\sum_{\mathbf{w} =1}^W n_{\mathbf{w}} \beta_{\mathbf{w}}c_{j,\odot n_1\cdots\odot n_W}^0(t,k_{\odot n_1}^{1},\cdots,k_{\odot n_W}^{W}), m=0.\label{Ddelaymodel3} 
\end{align}
\end{subequations}
\end{tiny}%

For the third component of the RHS of Eq.~(\ref{CjkmParaW}), we have
\begin{tiny}
\begin{equation} \label{con:Paracjkm3}
\begin{aligned}
&~~~~i\sum_{p=1}^{n_1+1}\int G(k_p,t) c_{j,k}^{m-1}(t,k_{1\odot n_1+1}^1,\cdots,k_{p\odot n_1+1}^1,\cdots, k_{n_1+1\odot n_1+1}^1;\\
&~~~~k_{1\odot n_2}^2,\cdots,k_{n_2\odot n_2}^2;\cdots; k_{1\odot n_W}^W;\cdots,k_{n_W\odot n_W}^W) \mathrm{d}k_p \\
&=i\sum_{p=1}^{n_1+1}\int G(k_p,t) \int_0^t \dot{c}_{j,k}^{m-1}(u,k_{1\odot n_1+1}^1,\cdots,k_{p\odot n_1+1}^1,\cdots, k_{n_1+1\odot n_1+1}^1;\\
&~~~~k_{1\odot n_2}^2,\cdots,k_{n_2\odot n_2}^2;\cdots; k_{1\odot n_W}^W;\cdots,k_{n_W\odot n_W}^W) \mathrm{d}u \mathrm{d}k_p\\
&=i\sum_{p=1}^{n_1+1}\int G(k_p,t) \int_0^t i\sum_{q=1}^{n_1+1} G^*(k_q,u) c_{j,k}^{m}(u,k_{1\odot n_1}^1,\cdots,k_{q-1\odot n_1}^1,\\
&~~~~k_{q+1\odot n_1}^1,\cdots, k_{n_1-1\odot n_1}^1; k_{1\odot n_2}^2,\cdots,k_{n_2\odot n_2}^2;\cdots;\\
&~~~~k_{1\odot n_W}^W;\cdots,k_{n_W\odot n_W}^W) \mathrm{d}u \mathrm{d}k_p\\
&=-\sum_{p=1}^{n_1+1}\int G(k_p,t) \int_0^t  G^*(k_p,u) c_{j,k}^{m}(u,k_{1\odot n_1}^1,\cdots,k_{p-1\odot n_1}^1,k_{p+1\odot n_1}^1,\\
&~~~~\cdots, k_{n_1-1\odot n_1}^1; k_{1\odot n_2}^2,\cdots,k_{n_2\odot n_2}^2;\cdots; k_{1\odot n_W}^W;\cdots,k_{n_W\odot n_W}^W) \mathrm{d}u \mathrm{d}k_p\\
&=-\sum_{p=1}^{n_1+1}\int_0^t \int G(k_p,t)   G^*(k_p,u) c_{j,k}^{m}(u,k_{1\odot n_1}^1,\cdots,k_{p-1\odot n_1}^1,k_{p+1\odot n_1}^1,\\
&~~~~\cdots, k_{n_1-1\odot n_1}^1; k_{1\odot n_2}^2,\cdots,k_{n_2\odot n_2}^2;\cdots; k_{1\odot n_W}^W;\cdots,k_{n_W\odot n_W}^W) \mathrm{d}k_p \mathrm{d}u \\
&=-\left|G_0\right|^2 \sum_{p=1}^{n_1+1}\int_0^t \int \sin^2(k_p L) e^{-i(\omega_p-\Delta_0)(t-u)} c_{j,k}^{m}(u,k_{1\odot n_1}^1,\cdots,k_{p-1\odot n_1}^1,\\
&~~~~k_{p+1\odot n_1}^1,\cdots, k_{n_1-1\odot n_1}^1; k_{1\odot n_2}^2,\cdots,k_{n_2\odot n_2}^2;\cdots;\\
&~~~~k_{1\odot n_W}^W;\cdots,k_{n_W\odot n_W}^W) \mathrm{d}k_p \mathrm{d}u \\
&=-\kappa n_1 \left [c_{j,k}^{m}(t,k_{\odot n_1}^{1},\cdots,k_{\odot n_W}^{W}) - e^{i\Delta_0\tau} c_{j,k}^{m}(t-\tau,k_{\odot n_1}^{1},\cdots,k_{\odot n_W}^{W})\right ],
\end{aligned}
\end{equation}
\end{tiny}%
where the derivation of the influence by the round trip delay $\tau$ is the same as in the single-waveguide circumstance in Eq.~(\ref{con:examplefourpart2}).

\begin{remark}
The dynamics of the states that all the generated photons are in the cavity (Eq.~(\ref{photoncavity})) is not influenced by the parallel waveguide architecture because the cavity is only coupled with the nearest waveguide.
\end{remark}

\subsection{Transmission property of photons in the parallel waveguides}
The transmission of photons between the cavity and the parallel waveguides is determined by the design of the feedback loop and the coupling strengths among the parallel waveguides. By tuning the length of the feedback loop, the photons can be controlled to be emitted or not emitted from the cavity into the waveguides.

\begin{mypro}
When $\Delta_0\tau = 2n\pi$ and $\delta_j =0$, there are no photons in the parallel waveguides.
\end{mypro}
\begin{proof}
When $\Delta_0\tau = 2n\pi$ and $\delta_j =0$, the photons cannot be transmitted from the cavity to the first waveguide according to Eq.~(\ref{photoncavity}) for the same reason as that in Proposition~\ref{photonstate}.  Thus, there are no photons in the parallel waveguides.
\end{proof}

Noticing that the dynamics of Eq.~(\ref{photoncavity}) is the same as that in Eq.~(\ref{con:LinearControlDelay}) based on the same initial condition that the atom is at the highest energy level and there are no photons in the cavity or waveguide,   when the system in Eq.~(\ref{con:LinearControlDelay}) is exponentially stable, then $\lim_{t\rightarrow \infty}c_j^j(t) = 0$ for arbitrary $j$ in Eq.~(\ref{photoncavity}) and all the emitted photons are in the parallel waveguide array. Therefore, we consider the following simplified model for the photonic states in the parallel waveguides as
\begin{small}
\begin{equation} \label{con:tlarget0}
\begin{aligned}
&\dot{c}_{j,\odot n_1\cdots\odot n_W}^0(t,k_{\odot n_1}^{1},\cdots,k_{\odot n_W}^{W}) =\\
&i\sum_{\mathbf{w} =1}^{W-1} K_{\mathbf{w} (\mathbf{w}+1)} c_{j,k}^{0}(t,k_{\odot n_1}^{1},\cdots,k_{\odot n_{\mathbf{w}}+1}^{\mathbf{w}},k_{\odot n_{\mathbf{w}+1}-1}^{\mathbf{w}+1},\cdots,k_{\odot n_W}^{W})\\
& + i\sum_{\mathbf{w} =2}^{W} K_{\mathbf{w}(\mathbf{w}-1)}c_{j,k}^{0}(t,k_{\odot n_1}^{1},\cdots,k_{\odot n_{\mathbf{w}-1}-1}^{\mathbf{w}-1},k_{\odot n_{\mathbf{w}}+1}^{\mathbf{w}},\cdots,k_{\odot n_W}^{W})\\
&+i\sum_{\mathbf{w} =1}^W n_{\mathbf{w}} \beta_{\mathbf{w}}c_{j,\odot n_1\cdots\odot n_W}^0(t,k_{\odot n_1}^{1},\cdots,k_{\odot n_W}^{W}).
\end{aligned}
\end{equation}
\end{small}%

In particular, when $N_W  = 1$, i.e., the single-excitation case,  denoting $c_{j,\odot n_1\cdots\odot n_W}^0(t,k_{\odot n_1}^{1},\cdots,k_{\odot n_W}^{W})$ by $c_{\mathbf{w}}(t)$ for the state that there is one photon of the given mode $k$ in the $\mathbf{w}$-th waveguide. Then Eq.~(\ref{con:tlarget0})  can be rewritten as
\begin{small}
\begin{equation} \label{con:OnephotonParaWave}
\begin{aligned}
\dot{c}_{\mathbf{w}}(t) = iK_{(\mathbf{w}-1) \mathbf{w}} c_{\mathbf{w}-1}(t) +  iK_{\mathbf{w} (\mathbf{w}+1)}c_{\mathbf{w}+1}(t) + i \beta_{\mathbf{w}}c_{\mathbf{w}}(t).
\end{aligned}
\end{equation}
\end{small}%
Denote $C_W(t) = [c_{1}(t), \cdots, c_{\mathbf{w}}(t), \cdots,c_{W}(t)]^{\top} \in \mathbf{C}^{W}$. Then
\[
\dot{C}_W(t)  = iG_W C_W(t),
\]
where
\begin{small}
\begin{equation}\label{con:triangle}
\begin{aligned}
  &G_W  =
   \begin{bmatrix}
    \beta_1  &K_{12}  & 0    & \cdots & 0  &0 \\
    K_{12}  &\beta_2  & K_{23} &\cdots    & 0  &0 \\
    0  &K_{23}  &\beta_3  &\cdots &0  &0 \\
    \vdots & \vdots  & \vdots  &\ddots &\vdots &\vdots\\
     0 & 0 & 0 &\cdots   &\beta_{W-1} &K_{(W-1)W} \\
     0 & 0 & 0 &\cdots   &K_{(W-1)W} &\beta_{W} \\
  \end{bmatrix}.
  \end{aligned}
\end{equation}
\end{small}%
As all the eigenvalues of the real symmetric matrix $G_W$ are real,  the eigenvalues of $iG_W$ must be purely imaginary. Thus $C_W(t)$ oscillates, i.e., the photon is superposed over the parallel waveguides.

\subsection{A simulation example with $N=2$, $W=3$}
Take the circumstance that a two-level atom is coupled with a cavity, and the cavity is coupled with three parallel waveguides. Let the initial state be $|1,0,0\rangle$. Then there is at most one photon in the cavity or the waveguides for all time.  The quantum state of the system can be represented as
\begin{small}
\begin{equation} \label{con:AtomCavityThreeW}
\begin{aligned}
|\Psi(t)\rangle &= c_e(t)|1,0,0\rangle + c_{g}(t)|0,1,0\rangle + \int c_k^1(t,k)|0,0,k_{\odot 1}^1\rangle \mathrm{d}k \\
&+ \int c_k^2(t,k)|0,0,k_{\odot 1}^2\rangle\mathrm{d}k + \int c_k^3(t,k)|0,0,k_{\odot 1}^3\rangle\mathrm{d}k,
\end{aligned}
\end{equation}
\end{small}%
where $|1,0,0\rangle$ means that the atom is excited and there are no photons in the cavity or waveguides, $|0,1,0\rangle$ means that the atom is at the ground state and there is one photon in the cavity, $|0,0,k_{\odot 1}^{\mathbf{w}}\rangle$ means that the atom is at the ground state, there is no photon in the cavity, and there is one photon in the $\mathbf{w}$-th waveguide. Additionally, we denote the time-varying amplitude of the central mode of the photon in the $\mathbf{w}$-th waveguide as $c_{\mathbf{w}}(t)$ with $\mathbf{w} = 1,2,3$. The simulation results when $G_0 = 0.25$, $\gamma =1$, $\Delta_0\tau = \pi$, $K_{12} = K_{23} = 0.5$ and $\beta_{1} = \beta_{2} = \beta_{3} = 0$ are shown in Fig.~\ref{fig:TwolevelThreeWaveguide}. The oscillation of the photons in waveguides is determined by the zero points of the following determinant according to Eq.~(\ref{con:triangle}) as
\begin{small}
\begin{equation} \label{con:ThreeWaveguidePole}
\begin{aligned}
|sI-iG_W| = s\left(s^2+K_{12}^2 + K_{23}^2\right ),
\end{aligned}
\end{equation}
\end{small}%
which reveals that the steady photon amplitudes oscillate around different nonzero mean values with the frequency determined by $\sqrt{K_{12}^2 + K_{23}^2}$ when $\beta_{\mathbf{w}} = 0$.

\begin{figure}[h]
\centerline{\includegraphics[width=0.9\columnwidth]{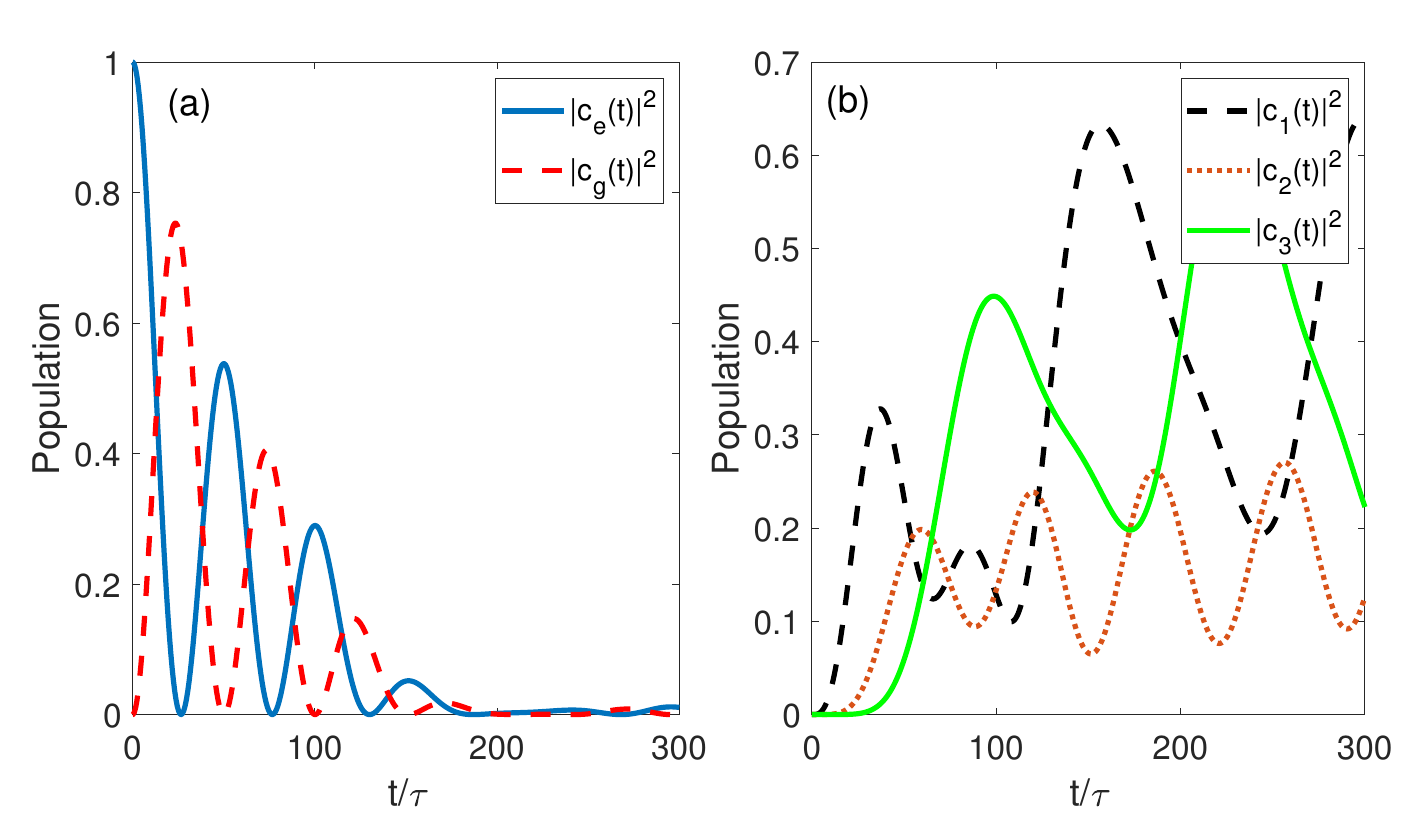}}
\caption{Coherent feedback control of a two-level atom in a cavity coupled with three parallel waveguides.}	\label{fig:TwolevelThreeWaveguide}
\end{figure}

\subsection{A simulation example with $N=3$, $W=2$}
Take the circumstance that a three-level ladder type atom is coupled with a cavity, and the cavity is coupled with two parallel waveguides. Then the state of the quantum system is
\begin{scriptsize}
\begin{equation} \label{con:three32}
\begin{aligned}
&|\Psi(t)\rangle = c_0^0(t)|2,0,0\rangle + c_{1}^1(t)|1,1,0\rangle + c_2^2(t)|0,2,0\rangle \\
& + \int c_{1,\odot 1\odot 0}^{0}(t,k_{1,1\odot 1}^1)|1,0,\{k_{\odot 1}^1\},\{k_{\odot 0}^2\}\rangle \mathrm{d}k_{1,1\odot 1}^1 \\
&+ \int c_{1,\odot 0\odot 1}^{0}(t,k_{1,1\odot 1}^2)|1,0,\{k_{\odot 0}^1\},\{k_{\odot 1}^2\}\rangle \mathrm{d}k_{1,1\odot 1}^2\\
&+ \int c_{2,\odot 1\odot 0}^{1}(t,k_{2,1\odot 1}^1)|0,1,\{k_{\odot 1}^1\},\{k_{\odot 0}^2\}\rangle \mathrm{d}k_{2,1\odot 1}^1 \\
&+ \int c_{2,\odot 0\odot 1}^{1}(t,k_{2,1\odot 1}^2)|0,1,\{k_{\odot 0}^1\},\{k_{\odot 1}^2\}\rangle \mathrm{d}k_{2,1\odot 1}^2\\
&+ \int \int c_{2,\odot 2\odot 0}^{0} (t,k_{2,1\odot 2}^1,k_{2,2\odot 2}^1)|0,0,\{k_{\odot 2}^1\},\{k_{\odot 0}^2\}\rangle \mathrm{d}k_{2,1\odot 2}^1 \mathrm{d}k_{2,2\odot 2}^1 \\
&+ \int \int c_{2,\odot 0\odot 2}^{0} (t,k_{2,1\odot 2}^2,k_{2,2\odot 2}^2)|0,0,\{k_{\odot 0}^1\},\{k_{\odot 2}^2\}\rangle \mathrm{d}k_{2,1\odot 2}^2 \mathrm{d}k_{2,2\odot 2}^2 \\
&+ \int \int c_{2,\odot 1\odot 1}^{0} (t,k_{2,1\odot 1}^1,k_{2,1\odot 1}^2)|0,0,\{k_{\odot 1}^1\},\{k_{\odot 1}^2\}\rangle \mathrm{d}k_{2,1\odot 1}^1 \mathrm{d}k_{2,1\odot 1}^2 ,\\
\end{aligned}
\end{equation}
\end{scriptsize}%
where the meaning of the state vectors are as follows:

(1)~$|2,0,0\rangle$: the atom is at the excited state $|2\rangle$ and there are not any photons in the cavity or waveguides;

(2)~$|1,1,0\rangle$: the atom is $|1\rangle$ and there is one photon in the cavity but no photons in the waveguides;

(3)~$|0,2,0\rangle$: the atom is at the ground state $|0\rangle$ and there are two photons in the cavity but no photons in the waveguides;

(4)~$|1,0,\{k_{\odot 1}^1\},\{k_{\odot 0}^2\}\rangle$: the atom is $|1\rangle$ and there is one photon in the first waveguide, and the mode of the photon is $k_{2,1\odot 1}^1$, where the subscript represent that there are overall one excitons and the photon is in the first waveguide;

(5)~$|1,0,\{k_{\odot 0}^1\},\{k_{\odot 1}^2\}\rangle$: the atom is $|1\rangle$ and there is one photon in the second waveguide;

(6)~$|0,1,\{k_{\odot 1}^1\},\{k_{\odot 0}^2\}\rangle$: the atom is $|0\rangle$, there is one photon in the cavity and one photon in the first waveguide;

(7)~$|0,1,\{k_{\odot 0}^1\},\{k_{\odot 1}^2\}\rangle$: the atom is $|0\rangle$, there is one photon in the cavity and one photon in the second waveguide;

(8)~$|0,0,\{k_{\odot 2}^1\},\{k_{\odot 0}^2\}\rangle$: the atom is $|0\rangle$, there are two photons in the first waveguide;

(9)~$|0,0,\{k_{\odot 0}^1\},\{k_{\odot 2}^2\}\rangle$: the atom is $|0\rangle$, there are two photons in the second waveguide;

(10)~$|0,0,\{k_{\odot 1}^1\},\{k_{\odot 1}^2\}\rangle$: the atom is $|0\rangle$, there is one photon in the first waveguide and the other in the second waveguide.

The meanings of the photon mode in the waveguide are as following:

(1)~$k_{1,1\odot 1}^1$: the mode of the photon in the first waveguide when $j=1$ and there is one photon in the first waveguide;

(2)~$k_{1,1\odot 1}^2$: the mode of the photon in the second waveguide when $j=1$ and there is one photon in the second waveguide;

(3)~$k_{2,1\odot 1}^1$: the mode of the photon in the first waveguide when $j=2$, there is one photon in the cavity and the other is in the first waveguide;

(4)~$k_{2,1\odot 1}^2$: the mode of the photon in the second waveguide when $j=2$, there is one photon in the cavity and the other is in the second waveguide;

(5)~$k_{2,1\odot 2}^1,k_{2,2\odot 2}^1$: the modes of the two photons in the first waveguide when $j=2$ and there are two photons in the first waveguide;

(6)~$k_{2,1\odot 2}^2,k_{2,2\odot 2}^2$: the modes of the two photons in the second waveguide when $j=2$ and there are two photons in the second waveguide;

(7)~$k_{2,1\odot 1}^1,k_{2,1\odot 1}^2$: when $j=2$ and there is one photon in the first and second waveguide, respectively, $k_{2,1\odot 1}^1$ is for the photon mode in the first waveguide, $k_{2,1\odot 1}^2$ for the photon in the second waveguide.

The population representing that there is one photon in the first or the second waveguide can be represented
according to Eq.~(\ref{con:Pwn}).
The dynamics of the quantum states in Eq.~(\ref{con:three32}) is governed by Eq.~(\ref{con:ParallelControl}) with $N=3$ and $W = 2$. In the following numerical simulations, $\Delta_0 =50$, $G_0 = 0.2$, $\gamma_1 = \gamma_2 =0.3$, $\Delta_0\tau = 3\pi$, $\beta_1 = \beta_2 = 0.1$, $K_{12} = K_{21} = 0.5$.
\begin{figure}[h]
\centerline{\includegraphics[width=1\columnwidth]{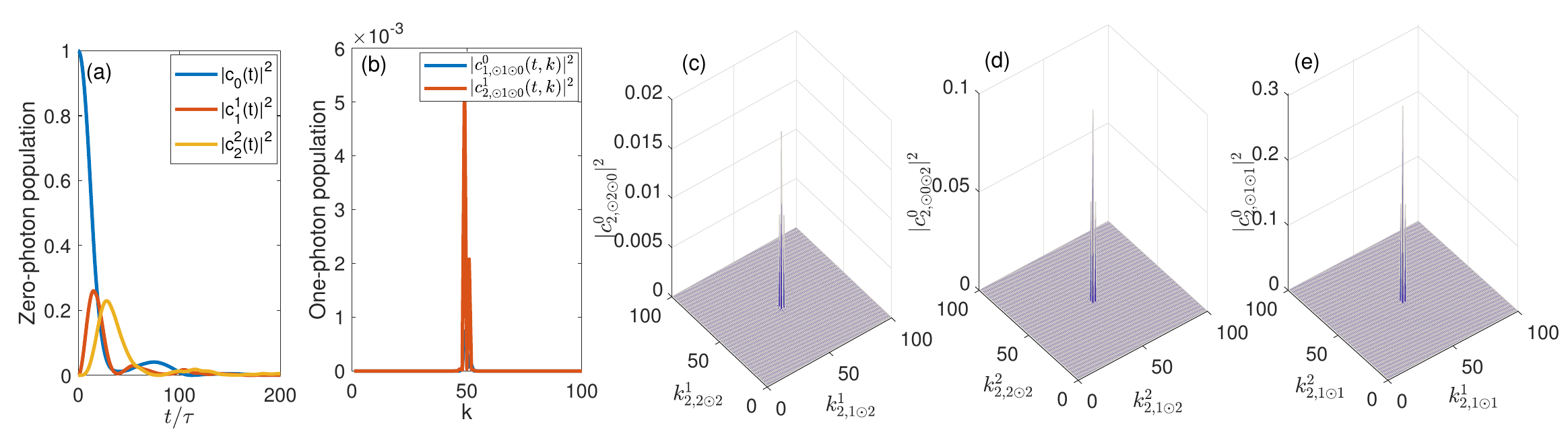}}
\caption{The populations of the zero-, one- and two-photon states when $N=3$, $W =2$.}
	\label{fig:threelevelTwowaveguide}
\end{figure}
As illustrated in Fig.~\ref{fig:threelevelTwowaveguide}, the excited three-level atom can emit two photons, and finally the two-photon state oscillates in the parallel waveguide array.

\section{Conclusion} \label{Sec:Conclusion}
In this paper, we have studied a coherent feedback scheme in the architecture that a multi-level atom is coupled with a cavity and the cavity is coupled with a single waveguide or multiple parallel waveguides. The number of photons in the waveguide is determined by the evolution of the quantum states, especially by their exponential stability. By tuning the length of the feedback loop, all the photons can be emitted into the waveguides, thus the amplitudes of eigenstates representing that there are photons in the cavity or the atom is excited  simultaneously converge to zero; or alternatively the multi-level atom exchanges photons with the cavity via Rabi oscillations and there are no photons in the waveguides. The photonic states in the parallel waveguides can be represented in the tensor format and the coupling parameters among the parallel waveguides can further influence the distribution of photons in the waveguide array. We have shown how the quantum control dynamics can be modeled as linear control systems with a round-trip delay, and the convergence of the eigenstates and the generation of photons can be analyzed by means of the stability theory for linear time-invariant or time-varying systems, both of which can be studied via Lyapunov-Krasovkii functions. This demonstrates a nice application of time-delayed linear control theory in quantum systems.

\begin{appendices}
\section{Derivation of the delay dependent control equation} \label{Sec:ModeldelayAppend}
In this Appendix, we derive the delay-dependent ODE system in Eq.~(\ref{con:controleq2}) in the main text.

The third term of the right-hand side of Eq.~(\ref{Nmodel2}) can be represented as
\begin{scriptsize}
\begin{equation} \label{con:delay0}
\begin{aligned}
&~~~~i\sum_{p=1}^{j-m+1}\int G(k_p,t) c_{j,k}^{m-1}(t,k_1,\cdots,k_{p-1},k_p,k_{p+1}\dots,k_{j-m+1}) \mathrm{d}k_p \\
&=i\sum_{p=1}^{j-m+1}\int  G_0\sin(k_pL)e^{-i(\omega_p-\Delta_0)t} \\
&~~~~c_{j,k}^{m-1}(t,k_1,\cdots,k_{p-1},k_p,k_{p+1}\dots,k_{j-m+1}) \mathrm{d}k_p \\
&=i\sum_{p=1}^{j-m+1}\int  G_0\sin(k_pL)e^{-i(\omega_p-\Delta_0)t} \\
&~~~~\int_0^t \dot{c}_{j,k}^{m-1}(u,k_1,\cdots,k_{p-1},k_p,k_{p+1}\dots,k_{j-m+1}) \mathrm{d}u \mathrm{d}k_p \\
&=i\sum_{p=1}^{j-m+1}\int \int_0^t   G_0\sin(k_pL)e^{-i(\omega_p-\Delta_0)t}\\
&~~~~\dot{c}_{j,k}^{m-1}(u,k_1,\cdots,k_{p-1},k_p,k_{p+1}\dots,k_{j-m+1}) \mathrm{d}u \mathrm{d}k_p,
\end{aligned}
\end{equation}
\end{scriptsize}%
where the initial condition $c_{j,k}^{m-1}(0,k_1,\cdots,k_{p-1},k_p,k_{p+1}\dots,k_{j-m+1})  =0$ has been used to get the second last step. Moreover, according to Eq. \eqref{Nmodel2},
\begin{scriptsize}
\begin{equation} \label{con:delay}
\begin{aligned}
&~~~~\dot{c}_{j,k}^{m-1}(u,k_1,\cdots,k_{p-1},k_p,k_{p+1}\dots,k_{j-m+1}) \\
&= i\sqrt{m-1}\gamma_{N-j}  c_{j-1,k}^{m-2}(u,k_1,\cdots,k_{j-m+1}) \\
&+ i\sqrt{m}\gamma_{N-j-1}  c_{j+1,k}^{m}(u,k_1,\cdots,k_{j-m+1})\\
&+i\sum_{p=1}^{j-m+2}\int G(k_p,u) c_{j,k}^{m-2}(u,k_1,\cdots,k_{p-1},k_p,k_{p+1}\dots,k_{j-m+2}) \mathrm{d}k_p \\
&+i\sum_{p=1}^{j-m+1} G^*(k_p,u) c_{j,k}^{m}(u,k_1,\cdots,k_{p-1},k_{p+1},\dots, k_{j-m}), m>0. \\
\end{aligned}
\end{equation}
\end{scriptsize}%

Firstly, let us look at the first term on the RHS of Eq.~\eqref{con:delay}. With $\tau = 2L/c$, then
\begin{scriptsize}
\begin{equation} \label{con:example}
\begin{aligned}
&~~~~ -\sqrt{m-1}\gamma_{N-j} G_0 \int\int_0^t  \sin(k_pL)e^{-i(\omega_p-\Delta_0)t}   \\
&~~~~c_{j-1,k}^{m-2}(u,k_1,\cdots,k_{j-m+1}) \mathrm{d}u \mathrm{d}k_p \\
&=  -\sqrt{m-1}\gamma_{N-j} G_0 \int\int_0^t  \frac{e^{ik_pL} - e^{-ik_p L}}{2i}e^{-i(\omega_p-\Delta_0)t}  \\
&~~~~c_{j-1,k}^{m-2}(u,k_1,\cdots,k_{j-m+1}) \mathrm{d}u \mathrm{d}k_p \\
&=\frac{i\sqrt{m-1}\gamma_{N-j} G_0}{2} \int\int_0^t  \left (e^{-i(\omega_p-\Delta_0)t + ik_pL} -e^{-i(\omega_p-\Delta_0)t - ik_pL}\right ) \\
&~~~~c_{j-1,k}^{m-2}(u,k_1,\cdots,k_{j-m+1}) \mathrm{d}u \mathrm{d}k_p \\
&=\frac{i\sqrt{m-1}\gamma_{N-j} G_0}{2}e^{i\Delta_0 t} \int \left ( e^{-i\omega_p(t-\frac{\tau}{2})} -e^{-i\omega_p(t+\frac{\tau}{2})}\right ) \\
&~~~~\int_0^t c_{j-1,k}^{m-2}(u,k_1,\cdots,k_{j-m+1}) \mathrm{d}u \mathrm{d}k_p. \\
\end{aligned}
\end{equation}
\end{scriptsize}

Denote $\tilde{c}(t,k_p) = \int_0^t c_{j-1,k}^{m-2}(u,k_1,\cdots,k_{j-m+1}) \mathrm{d}u$, then Eq.~(\ref{con:example}) reads
\begin{scriptsize}
\begin{equation} \label{con:example2}
\begin{aligned}
&~~~~\frac{i\sqrt{m-1}\gamma_{N-j} G_0}{2} e^{i\Delta_0 t} \int \left( e^{-i\omega_p\left(t-\frac{\tau}{2}\right)} -e^{-i\omega_p\left(t+\frac{\tau}{2}\right)}\right) \tilde{c}(t,k_p)\mathrm{d}k_p \\
&=\frac{i\sqrt{m-1}\gamma_{N-j} G_0}{2} e^{i\Delta_0 t} \left[\delta\left(t-\frac{\tau}{2}\right) -\delta\left(t+\frac{\tau}{2}\right)\right]\\
&~~~~\int_0^t c_{j-1,k}^{m-2}(u,k_1,\cdots,k_{j-m+1}) \mathrm{d}u \\
&~~~~- \frac{i\sqrt{m-1}\gamma_{N-j} G_0}{2} e^{i\Delta_0 t} \left[\delta\left(t-\frac{\tau}{2}\right) -\delta\left(t+\frac{\tau}{2}\right)\right] \\
&~~~~\int \int_0^t \frac{\partial c_{j-1,k}^{m-2}(u,k_1,\cdots,k_{j-m+1})}{\partial k_p} \mathrm{d}u \mathrm{d}k_p\\
&=0,
\end{aligned}
\end{equation}
\end{scriptsize}%
because the integration equals zero when $t\neq \tau/2$ and the amplitude is continuous in the time domain.
Similarly the second term on the RHS  of Eq.~(\ref{con:delay}) equals zero.

Substituting the third component of Eq.~(\ref{con:delay}) into Eq.~(\ref{con:delay0}), the integration reads
\begin{footnotesize}
\begin{equation} \label{con:examplethreepart}
\begin{aligned}
&-\left|G_0\right|^2 \sum_{p=1}^{j-m+1}\int \sin(k_pL)e^{-i(\omega_p-\Delta_0)t}
\int_0^t  \sum_{q=1}^{j-m+2}
\int \sin(k_qL)\\
&e^{-i(\omega_q-\Delta_0)t}  c_{j,k}^{m-2}(u,k_1,\cdots,k_{q-1},k_q,k_{q+1}\dots,k_{j-m+2}) \mathrm{d}k_q.
\end{aligned}
\end{equation}
\end{footnotesize}%

When $p\neq q$, the above integration equals zero. When $p=q$, by the calculations in Eq.~(\ref{con:example2})
\begin{scriptsize}
\begin{equation} \label{con:examplethreepart2}
\begin{aligned}
&~~~~-\left|G_0\right|^2 \sum_{p=1}^{j-m+1}\int \sin(k_pL)e^{-i(\omega_p-\Delta_0)t}  \int_0^t \int \sin(k_pL)e^{-i(\omega_p-\Delta_0)t} \\
&~~~~c_{j,k}^{m-2}(u,k_1,\cdots,k_{p-1},k_p,k_{p+1}\dots,k_{j-m+2}) \mathrm{d}k_p \mathrm{d}u \mathrm{d}k_p \\
&=-\left|G_0\right|^2  \sum_{p=1}^{j-m+1}\int \int \sin^2(k_pL)e^{-i2(\omega_p-\Delta_0)t}  \\
&\int_0^t  c_{j,k}^{m-2}(u,k_1,\cdots,k_{p-1},k_p,k_{p+1}\dots,k_{j-m+2}) \mathrm{d}u \mathrm{d}k_p \mathrm{d}k_p \\
&=-\left|G_0\right|^2 \sum_{p=1}^{j-m+1}\int \int \sin^2(k_pL)e^{-i2(\omega_p-\Delta_0)t} \bar{c}(t,k_p) \mathrm{d}k_p \mathrm{d}k_p \\
&=-\left|G_0\right|^2 \sum_{p=1}^{j-m+1}\int \int \left (\frac{1}{2} -\frac{1}{4}e^{i\omega_p\tau} -\frac{1}{4}e^{-i\omega_p\tau}\right )e^{-i2(\omega_p-\Delta_0)t} \\
&~~~~\bar{c}(t,k_p) \mathrm{d}k_p \mathrm{d}k_p\\
&=0.
\end{aligned}
\end{equation}
\end{scriptsize}%
Substituting the fourth term on the RHS of Eq.~(\ref{con:delay}) into Eq.~(\ref{con:delay0}), the integration reads
\begin{scriptsize}
\begin{equation} \label{con:examplefourpart2}
\begin{aligned}
&~~~~-\sum_{p=1}^{j-m+1}\int \int_0^t   G_0\sin(k_pL)e^{-i(\omega_p-\Delta_0)t} \sum_{q=1}^{j-m+1} G^*(k_q,u) \\
&~~~~c_{j,k}^{m}(u,k_1,\cdots,k_{q-1},k_{q+1},\dots, k_{j-m}) \mathrm{d}u \mathrm{d}k_p \\
&=-\left|G_0\right|^2 \sum_{p=1}^{j-m+1}\int_0^t \int \sin^2(k_pL) e^{-i(\omega_p-\Delta_0)(t-u)} \\
&~~~~c_{j,k}^{m}(u,k_1,\cdots,k_{p-1},k_{p+1},\dots, k_{j-m})  \mathrm{d}k_p \mathrm{d}u\\
&=-\frac{\left|G_0\right|^2}{4c} \sum_{p=1}^{j-m+1}\int_0^t \int \left [2e^{-i(\omega_p-\Delta_0)(t-u)} - e^{i\Delta_0\tau} e^{-i(\omega_p-\Delta_0)(t-u-\tau)} \right.\\
&~~~~\left.- e^{-i\Delta_0\tau}e^{-i(\omega_p-\Delta_0)(t-u+\tau)}\right ] \\
&~~~~c_{j,k}^{m}(u,k_1,\cdots,k_{p-1},k_{p+1},\dots, k_{j-m}) \mathrm{d}k_p \mathrm{d}u\\
&=-\frac{\left|G_0\right|^2}{4c} \sum_{p=1}^{j-m+1}\int_0^t \left [2\delta(t-u) - e^{i\Delta_0\tau}\delta(t-u-\tau) \right.\\
&~~~~\left.- e^{-i\Delta_0\tau}\delta(t-u+\tau)\right ]c_{j,k}^{m}(u,k_1,\cdots,k_{p-1},k_{p+1},\dots, k_{j-m})  \mathrm{d}u\\
&=-\frac{\left|G_0\right|^2}{4c} \sum_{p=1}^{j-m+1} \left [c_{j,k}^{m}(t,k_1,\cdots,k_{p-1},k_{p+1},\dots, k_{j-m}) \right.\\
&~~~~\left.- e^{i\Delta_0\tau} c_{j,k}^{m}(t-\tau,k_1,\cdots,k_{p-1},k_{p+1},\dots, k_{j-m}) \right],
\end{aligned}
\end{equation}
\end{scriptsize}%
which is the component with delay in Eq.~(\ref{delaymodel2}) in the main text.

\section{Proof of \textbf{Proposition}~\ref{TACLyaUncertainty}} \label{Sec:ProofLyaApend}
In this Appendix, we prove \textbf{Proposition}~\ref{TACLyaUncertainty}.
Consider the Lyapunov-Krasovkii function
\begin{small}
\begin{equation} \label{con:LyaVxt}
\begin{aligned}
V(\tilde{X}_t) = \tilde{X}^\top(t)\tilde{P}\tilde{X}(t) + \int_{-\tau}^0 \tilde{X}^\top  (t+\theta) e^{2\beta\theta} \tilde{Q}\tilde{X}(t+\theta)  \mathrm{d}\theta,
  \end{aligned}
\end{equation}
\end{small}%
where $\tilde{P}$ and $\tilde{Q}$ are real positive-definite matrices.  Denote
\begin{numcases}{}
\tilde{\alpha}_1 = \lambda_{min}(\tilde{P}),\\
\tilde{\alpha}_2 = \lambda_{max}(\tilde{P}) + \tau \lambda_{max}(\tilde{Q}).
\end{numcases}%
Then clearly,
\begin{small}
\begin{equation}\label{con:VxtInequ}
   V(\tilde{X}_t) \geq \tilde{\alpha}_1 \| \tilde{X}_t \|^2 .
\end{equation}
\end{small}%
Notice that $V(\tilde{X}_t)$ in Eq. \eqref{con:LyaVxt} can be re-written as
\begin{small}
\begin{equation} \label{con:VxtCal00}
\begin{aligned}
V(\tilde{X}_t)  =& \begin{bmatrix}
\tilde{X}(t)\\
\tilde{X}(t-\tau)
\end{bmatrix}^\top   \mathcal{N}(\tilde{P})\begin{bmatrix}
\tilde{X}(t)\\
\tilde{X}(t-\tau)
\end{bmatrix} \\
&+ \int_{-\tau}^0 \tilde{X}^\top  (t+\theta) e^{2\beta\theta} \tilde{Q}\tilde{X}(t+\theta)  \mathrm{d}\theta,
\end{aligned}
\end{equation}
\end{small}%
where the matrix $\mathcal{N}$ is that defined in Eq.~(\ref{MNdef}). Differentiating both sides of Eq.~\eqref{con:VxtCal00} yields
\begin{small}
\begin{equation} \label{con:VxtCal}
\begin{aligned}
\frac{\mathrm{d}}{\mathrm{d}t} V(\tilde{X}_t) = & \tilde{X}^\top  (t)\tilde{P}\dot{\tilde{X}}(t) + \dot{\tilde{X}}(t) ^\top  \tilde{P}\tilde{X}(t)\\
&+  \frac{\mathrm{d}}{\mathrm{d}t} \int_{-\tau}^0 \tilde{X}^\top  (t+\theta) e^{2\beta\theta} \tilde{Q}\tilde{X}(t+\theta)  \mathrm{d}\theta.
\end{aligned}
\end{equation}
\end{small}%
Denote $u =t+\theta$ and  $F(t) = \int_{-\tau}^0 \tilde{X}^\top  (t+\theta) e^{2\beta\theta} \tilde{Q}\tilde{X}(t+\theta)  \mathrm{d}\theta= \int_{t-\tau}^t \tilde{X}^\top  (u) e^{2\beta(u-t)} \tilde{Q}\tilde{X}(u)  \mathrm{d}u$. We have
\begin{small}
\begin{equation} \label{con:VxtCal2}
\begin{aligned}
&~~~~\frac{\mathrm{d}}{\mathrm{d}t} \int_{-\tau}^0 \tilde{X}^\top  (t+\theta) e^{2\beta\theta} \tilde{Q}\tilde{X}(t+\theta)  \mathrm{d}\theta\\
&= -2\beta\int_{t-\tau}^t \tilde{X}^\top  (u) e^{2\beta(u-t)} \tilde{Q}\tilde{X}(u)  \mathrm{d}u \\
&~~~~+ \tilde{X}^\top  (t) \tilde{Q}\tilde{X}(t) - \tilde{X}^\top  (t-\tau) e^{-2\beta\tau} \tilde{Q}\tilde{X}(t-\tau)\\
&=-2\beta \int_{-\tau}^0 \tilde{X}^\top  (t+\theta) e^{2\beta\theta} \tilde{Q}\tilde{X}(t+\theta)  \mathrm{d}\theta \\
&~~~~+ \tilde{X}^\top  (t) \tilde{Q}\tilde{X}(t) - \tilde{X}^\top  (t-\tau) e^{-2\beta\tau} \tilde{Q}\tilde{X}(t-\tau).
\end{aligned}
\end{equation}
\end{small}%
Substituting Eq.~\eqref{con:VxtCal2} into Eq.~\eqref{con:VxtCal} yields
\begin{footnotesize}
\begin{equation} \label{con:VxtCal3}
\begin{aligned}
&~~~~\frac{\mathrm{d}}{\mathrm{d}t} V(\tilde{X}_t) \\
&= \tilde{X}^\top  (t)\tilde{P}\dot{\tilde{X}}(t) + \dot{\tilde{X}}^\top  (t) \tilde{P}\tilde{X}(t)\\
&~~~~-2\beta \int_{-\tau}^0 \tilde{X}^\top  (t+\theta) e^{2\beta\theta} \tilde{Q}\tilde{X}(t+\theta)  \mathrm{d}\theta + \tilde{X}^\top  (t) \tilde{Q}\tilde{X}(t) \\
&~~~~ - \tilde{X}^\top  (t-\tau) e^{-2\beta\tau} \tilde{Q}\tilde{X}(t-\tau)\\
&= \tilde{X}^\top  (t)\tilde{P} \left[(\tilde{A}_0 + \Upsilon(t)) \tilde{X}(t) + \tilde{B} \tilde{X}(t-\tau)\right] \\
&~~~~ + \left[(\tilde{A}_0 + \Upsilon(t)) \tilde{X}(t)
 + \tilde{B} \tilde{X}(t-\tau)\right]^\top   \tilde{P}\tilde{X}(t)\\
&~~~~-2\beta \int_{-\tau}^0 \tilde{X}^\top  (t+\theta) e^{2\beta\theta} \tilde{Q}\tilde{X}(t+\theta)  \mathrm{d}\theta + \tilde{X}^\top  (t) \tilde{Q}\tilde{X}(t) \\
&~~~~ - \tilde{X}^\top  (t-\tau) e^{-2\beta\tau} \tilde{Q}\tilde{X}(t-\tau)\\
&=\begin{bmatrix}
\tilde{X}(t)\\
\tilde{X}(t-\tau)
\end{bmatrix}^\top   \left \{\begin{bmatrix}
   \tilde{P} \tilde{A}_0 +\tilde{A}^\top  _0\tilde{P}+\tilde{Q}  & \tilde{P}\tilde{B}\\
   \tilde{B}^\top  \tilde{P} & -e^{-2\beta\tau}\tilde{Q}
  \end{bmatrix} \right. \\
 &~~~~\left. + \begin{bmatrix}
\Upsilon(t)^\top  \\
\mathbf{0}_{2(N-1)*1}
\end{bmatrix} \tilde{P} [I ~~~~\mathbf{0}_{2(N-1)*1}] \right.\\
&\left. ~~~~+\begin{bmatrix}
I\\
\mathbf{0}_{2(N-1)*1}
\end{bmatrix} \tilde{P} [\Upsilon(t) ~~~~\mathbf{0}_{2(N-1)*1}] \right \}\begin{bmatrix}
\tilde{X}(t)\\
\tilde{X}(t-\tau)
\end{bmatrix}\\
&~~~~-2\beta \int_{-\tau}^0 \tilde{X}^\top  (t+\theta) e^{2\beta\theta} \tilde{Q}\tilde{X}(t+\theta)  \mathrm{d}\theta + \tilde{X}^\top  (t) \tilde{Q}\tilde{X}(t) \\
&~~~~ - \tilde{X}^\top  (t-\tau) e^{-2\beta\tau} \tilde{Q}\tilde{X}(t-\tau)\\
&\leq \begin{bmatrix}
\tilde{X}(t)\\
\tilde{X}(t-\tau)
\end{bmatrix}^\top    \left (\begin{bmatrix}
   \tilde{P} \tilde{A}_0 +\tilde{A}^\top  _0\tilde{P}+\tilde{Q}  & \tilde{P}\tilde{B}\\
   \tilde{B}^\top  \tilde{P} & -e^{-2\beta\tau}\tilde{Q}
  \end{bmatrix} \right.\\
  &\left. ~~~~ + 2  \lambda_{max}(\tilde{P}) \|\Upsilon(t,\gamma_1,\cdots,\gamma_{N-1},\delta_1,\cdots,\delta_{N-1}) \|_2 \right ) \begin{bmatrix}
\tilde{X}(t)\\
\tilde{X}(t-\tau)
\end{bmatrix}\\
&~~~~-2\beta \int_{-\tau}^0 \tilde{X}^\top  (t+\theta) e^{2\beta\theta} \tilde{Q}\tilde{X}(t+\theta)  \mathrm{d}\theta + \tilde{X}^\top  (t) \tilde{Q}\tilde{X}(t)\\
&~~~~- \tilde{X}^\top  (t-\tau) e^{-2\beta\tau} \tilde{Q}\tilde{X}(t-\tau).
\end{aligned}
\end{equation}
\end{footnotesize}%
Consequently,
\begin{small}
\begin{equation} \label{con:VxtCal4}
\begin{aligned}
&~~~~\frac{\mathrm{d}}{\mathrm{d}t} V(\tilde{X}_t) + 2\beta   V(\tilde{X}_t) \\
&\leq
\begin{bmatrix}
\tilde{X}(t)\\
\tilde{X}(t-\tau)
\end{bmatrix}^\top    \left (\begin{bmatrix}
   \tilde{P} \tilde{A}_0 +\tilde{A}^\top  _0\tilde{P}+\tilde{Q}  & \tilde{P}\tilde{B}\\
   \tilde{B}^\top  \tilde{P} & -e^{-2\beta\tau}\tilde{Q}
  \end{bmatrix}  \right.\\
  &~~~~\left. + 2  \lambda_{max}(\tilde{P}) \|\Upsilon(t,\gamma_1,\cdots,\gamma_{N-1},\delta_1,\cdots,\delta_{N-1}) \|_2 \right.\\
&~~~~\left.  +2\beta \mathcal{N}(\tilde{P}) \right )
 \begin{bmatrix}
\tilde{X}(t)\\
\tilde{X}(t-\tau)
\end{bmatrix}.
\end{aligned}
\end{equation}
\end{small}%
Once the condition in Eq.~(\ref{con:LyatimeVary}) is satisfied, one has
\[
\frac{\mathrm{d}}{\mathrm{d}t} V(\tilde{X}_t) + 2\beta   V(\tilde{X}_t) \leq 0,
\]
and this, together with Eq.~\eqref{con:VxtInequ}, gives
\[
\tilde{\alpha}_1 \| \tilde{X}_t \|^2 \leq V(\tilde{X}_t) \leq e^{-2\beta t} V (\tilde{X}_0 ),
\]
where $V(\tilde{X}_0) = \tilde{X}^\top  (0)\tilde{P}\tilde{X}(0) + \int_{-\tau}^0 \tilde{X}^\top  (\theta) e^{2\beta\theta} \tilde{Q}\tilde{X}(\theta)  \mathrm{d}\theta$. Thus
\begin{scriptsize}
\begin{equation}
\begin{aligned}
 V(\tilde{X}_t) &\leq e^{-2\beta t} \lambda_{max}(\tilde{P}) \tilde{X}^\top  (0)\tilde{X}(0) + \int_{-\tau}^0 \tilde{X}^\top  (\theta) e^{2\beta(\theta-t)} \tilde{Q}\tilde{X}(\theta)  \mathrm{d}\theta \\
& \leq e^{-2\beta t} \lambda_{max}(\tilde{P})  |\tilde{X}_t |_{\tau}^{*^2} + e^{-2\beta t} \tau \lambda_{max}(\tilde{Q}) |\tilde{X}_t |_{\tau}^{*^2} \\
&=  \tilde{\alpha}_2 e^{-2\beta t}  |\tilde{X}_t |_{\tau}^{*^2},
\end{aligned}
\end{equation}
\end{scriptsize}%
where $|\tilde{X}_t |_{\tau}^* = \max_{t\in[-\tau,0]}\{\|\tilde{\varphi}(t)\|\}$, as defined in \textbf{Definition}~\ref{defstable}. That is, the system is exponentially stable according to \textbf{Definition}~\ref{defstable}.

\end{appendices}

\bibliographystyle{IEEEtran}
\bibliography{Nlevel}
%
\end{document}